%% file: main.tex
\documentclass{article}

\usepackage{arxiv}

\usepackage[utf8]{inputenc} 
\usepackage[T1]{fontenc}    
\usepackage{hyperref}       
\usepackage{url}            
\usepackage{booktabs}       
\usepackage{amsfonts}       
\usepackage{nicefrac}       
\usepackage{microtype}      
\usepackage{lipsum}

\usepackage{authblk}
\input{macros}

\title{Causal Relational Learning\thanks{This is an extended version of a paper that accepted for publication at the Proceedings of the 2020 International Conference on Management of Data \cite{carl2020}.}}

\begin{document}

\author[1 ]{Babak Salimi}
\author[2]{Harsh Parikh}
\author[1]{Moe Kayali}
\author[2]{Sudeepa Roy}
\author[3]{Lise Getoor}
\author[1]{Dan Suciu}

\affil[1]{ University of Washington}
\affil[2]{ Duke University}
\affil[3]{ University of California at Santa Cruz}
\maketitle

\begin{abstract}
\input{abstract.tex}
\end{abstract}

\input{introduction2.tex}
\input{Background2.tex}

\input{Framework.tex}

\input{semantics.tex}
\input{queryanswering.tex}

\input{Experiment.tex}

\input{Conclusion.tex}
\clearpage 
\bibliographystyle{plain}
\bibliography{Ref}

\clearpage
\input{proof.tex}
\end{document}

%% file: macros.tex
\usepackage[pdf]{graphviz}
\usepackage{dot2texi}
\makeatletter
\@ifundefined{verbatim@out}{\newwrite\verbatim@out}{}
\makeatother
\usepackage{amsfonts}
\usepackage{booktabs}
\usepackage{dsfont}
\usepackage{siunitx}
\usepackage{multirow}
\usepackage{graphicx}
\usepackage{listings}
\usepackage{commath}
\usepackage{epstopdf}
\usepackage{hyperref}
\usepackage{xcolor}
\usepackage{url}
\usepackage{mathtools}
\usepackage{caption}
\usepackage{alltt}
\usepackage{mathrsfs}
\usepackage{float}
\usepackage{subfig}
\usepackage[normalem]{ulem}
\usepackage{graphicx,fancyvrb}
\usepackage[linesnumbered,ruled,vlined]{algorithm2e}
\usepackage{algpseudocode}
\usepackage{amsmath,amssymb}
\usepackage{booktabs}
\usepackage{caption}
\usepackage{float}
\usepackage{capt-of}
\usepackage{booktabs}
\usepackage{colortbl}
\usepackage{xfrac}
\usepackage{footnote}
\usepackage{array}
\usepackage{arydshln}
\usepackage{epstopdf}
\DeclareGraphicsExtensions{.png,.pdf}
\usepackage{amsmath,amsfonts,amssymb}

\newcommand{\babak}[1]{{\texttt{\color{red} Babak: [{#1}]}}}

\definecolor{ochre}{HTML}{C06000}

\newcommand{\red}[1]{{\textbf {{\color{red} {#1}}}}}

\newcommand{\sys}{\textsc{C}a\textsc{RL}}

\newcommand{\outc}{Y}

\newcommand{\cg}{ G}

\newcommand{\Dom}{Dom}
\newcommand{\nbs}{\mathbb P}
\newcommand{\arbtreat}{\mathbb S}
\newcommand{\nbsvec}{\trevec}
\usepackage{array}

\newcommand{\bx}{ X}
\newcommand{\Scm}{\mb S}
\newcommand{\Scmdef}{(\Pred,\FAtt)}
\newcommand{\Rel}{\mb R}

\newcommand{\SelfPay}{\text{SelfPay}}
\newcommand{\Ethnicity}{\text{Eth}}
\newcommand{\Religion}{\text{Religion}}
\newcommand{\Gender}{\text{Sex}}
\newcommand{\Patient}{\text{Pa}}
\newcommand{\data}{\textsc{ReviewData}}
\newcommand{\DrugDose}{\text{Dose}}
\newcommand{\Diagnosis}{\text{Diag}}
\newcommand{\Severity}{\text{Severe}}
\newcommand{\Prescribed}{\text{Drug}}
\newcommand{\Caregiver}{\text{Care}}
\newcommand{\los}{\text{Len}}
\newcommand{\Death}{\text{Death}}
\newcommand{\Administered}{\text{Given}}
\newcommand{\Doctor}{\text{Doc}}
\newcommand{\sdata}{\textsc{Synthetic ReviewData}}
\newcommand{\mimicdata}{\textsc{MIMIC-III}}
\newcommand{\nisdata}{\textsc{NIS}}
\newcommand{\rel}{R}
\newcommand{\pred}{P}
\newcommand{\Pred}{\mb P}
\newcommand{\Ent}{\mb E}
\newcommand{\ent}{ E}

\newcommand{\FAtt}{{\mb A}}
\newcommand{\fAtt}{{\tt A}}

\newcommand{\FAttrobs}{{\FAtt}_{Obs}}

\newcommand{\unit}{\mathbb{U}}
\newcommand{\runit}{\unit_{\outc}}
\newcommand{\tunit}{\unit_{\treat}}
\newcommand{\xunit}{\unit_{\FAtt}}
\newcommand{\uunit}{\unit_{(\treat,\outc)}}

\newcommand{\skl}{\mb \Delta}

\newcommand{\rse}{\phi}

\newcommand{\rcr}{rule} 
\newcommand{\rcm}{\Phi}
\newcommand{\gcg}{G(\rcm^{\skl})}
\newcommand{\egcg}{\hat{G}(\rcm^{\skl})}

\newcommand{\Do}{{\tt do}}

\newcommand{\sem}{\Leftarrow}

\newcommand{\cm}{{{ M }}}

\newcommand{\Emb}{\Psi}
\newcommand{\emb}{{{ \psi }}}

\newcommand{\treat}{{{T}}}

\newcommand{\utable}{\mb D}

\newcommand{\Hidb}{\Emb}
\newcommand{\Hid}{\emb}
\newcommand{\ate}{ \textsc{ATE}}
\newcommand{\trevec}{\vec{t}}

\newcommand{\tre}{t}

\DeclarePairedDelimiterX{\infdivx}[2]{(}{)}{%
	#1\;\delimsize|\delimsize|\;#2%
}

\definecolor{brown}{rgb}{0.8,0.1,0.1}
\definecolor{BROWN}{rgb}{0.8,0.1,0.1}
\definecolor{red}{rgb}{1,0,0}
\definecolor{RED}{rgb}{1,0,0}

\newcommand{\reva}[1]{{{{#1}}}}
\newcommand{\revb}[1]{{{#1}}}
\newcommand{\revc}[1]{{{{#1}}}}
\definecolor{dartmouthgreen}{rgb}{0.05, 0.5, 0.06}
\definecolor{burntumber}{rgb}{0.54, 0.2, 0.14}
\definecolor{BURNTUMBER}{rgb}{0.54, 0.2, 0.14}
\newcommand{\revd}[1]{{{{#1}}}}
\newcommand{\revm}[1]{{{{#1}}}}

\newcommand{\Ex}{{{ { \mathbb E}}}}

\definecolor{mygreen}{rgb}{0,0.6,0}
\definecolor{mygray}{rgb}{0.5,0.5,0.5}
\definecolor{mymauve}{rgb}{0.58,0,0.82}

\usepackage{listings}
\lstnewenvironment{VerbatimText}[1][]{
    
    \lstset{fancyvrb=true,basicstyle=\footnotesize,captionpos=b,xleftmargin=2em,#1}
}{}

\usepackage{booktabs}
\usepackage{pgfplots}
\usepackage{pgfplotstable}

\newcommand{\Pa}{\mb{Pa}}

\newcommand{\pa}{\mb{pa}}

\DeclareMathOperator{\EX}{\mathbb{E}}

\usepackage{amsmath}

\newcommand{\mc}[1]{\mathcal{#1}}
\newcommand{\tsc}[1]{\textsc{#1}}

\newcommand{\ignore}[1]{}

\newcommand{\defeq}{\stackrel{\text{def}}{=}}
\newcommand*{\rom}[1]{\expandafter\@slowromancap\romannumeral #1@}

\newcommand{\indep}{\mbox{$\perp\!\!\!\perp$}}

\usepackage{esvect}

\newcommand{\vo}{\vec{o}\@ifnextchar{^}{\,}{}}

\newcommand{\aoe}{ {\tt \textsc{AOE}}}

\newcommand{\aide}{ {\tt \textsc{AIE}}}

\newcommand{\arlf}{ {\tt \textsc{ARE}}}

\newcommand{\AGG}{{\tt \textsc{AGG}}}

\newcommand{\pr}{{\tt {Pr}}}

\newcommand{\mb}[1]{{\bf {#1}}}
\colorlet{lightgray}{gray!20}

\newtheorem{defn}{Definition}[section]

\newtheorem{example}[defn]{Example}
\newtheorem{definition}[defn]{Definition}
\newtheorem{theorem}[defn]{Theorem}
\newtheorem{proof}[defn]{Proof}
\newtheorem{proposition}{Proposition}[section]
\newtheorem{observation}{Observation}[section]

\newcommand{\RNum}[1]{\uppercase\expandafter{\romannumeral #1\relax}}
\newcommand{\ie}{{\em i.e.}} 
\newcommand{\eg}{{\em e.g.}} 

\newcommand{\proj}[1]{{\Pi}}
\newcommand{\sel}[1]{{\sigma}}

\newcommand{\cut}[1]{}
\newcommand{\cutr}[1]{} 
\newcommand{\eat}[1]{}

%% file: abstract.tex

%
%

\revm{Causal inference is at the heart of empirical research in  natural and social sciences and is critical for scientific discovery and informed decision making. The gold standard in causal inference is performing 
{\em randomized controlled trials}; unfortunately these are not always feasible due to ethical, legal, or cost constraints. As an alternative, 
methodologies for causal inference from {\em observational data} have been developed in statistical studies and social sciences\cutr{\cite{rubin1970thesis,rubin2008comment,pearl2018book}}.} 
However, existing methods critically rely \revm{on restrictive assumptions such as} the study population consisting of {\em homogeneous elements} that can be represented in a single flat table, where each row is referred to as a {\em unit}. In contrast, in many real-world settings, the study domain naturally consists of 
{\em heterogeneous elements} with complex relational structure, where the data is naturally represented in multiple related tables.
In this paper, we present a formal framework for causal inference from such relational data.   We propose a declarative language called \sys\ for capturing causal background knowledge and assumptions and specifying  causal queries using simple Datalog-like rules. \revm{ \sys\ provides a foundation for inferring causality and reasoning about the effect of complex interventions in relational domains.}   
%
\revm{
We present an extensive experimental evaluation on real relational data to illustrate the applicability of \sys\ in
social sciences and healthcare.}

%% file: introduction2.tex
\section{Introduction}
\label{sec:intro}

The importance of causal inference for making informed policy
decisions has long been recognised in health, medicine, social
sciences, and other domains. However, today's decision-making systems
typically do not go beyond {\em predictive analytics} and thus fail to
answer questions such as ``What 
\revm{would} happen to revenue if the price of
X is lowered?'' While predictive analytics has achieved remarkable
success in diverse applications, it is mostly restricted to fitting a model to
observational data based on associational
patterns~\cite{pearl2018book}. Causal inference, on the other hand,
goes beyond associational patterns to the process that generates the
data, {thereby enabling analysts to reason about {\em interventions}  (e.g., ``Would requiring flu shots in schools reduce the chance of a future flu epidemic?") and {\em counterfactuals} (e.g., ``What would have happened if past flu shots were not taken?").} \revm{This adds significantly more information in data analysis compared to simple correlation or regression analysis; e.g., 
as the number of flu cases increases, the rate of
flu shots might also increase, but that does not imply that 
giving flu shots increases the spread of flu. 
This emphasizes the 
common saying that ``correlation is not causation'', which is known to all, but is easy to overlook if one is not careful while analyzing data for insights and possible actions.}
 
The gold standard in causal analysis is performing {\em randomized
  controlled trials}, where the {\em subjects} or {\em units} of
study are assigned randomly to a treatment or a \revm{control  (i.e., withheld
from the treatment) group}. The difference between the distribution of the outcome variable of
the treated and control groups represents the {\em causal effect} of the treatment on outcome. However,
control experiments are not always feasible due to ethical, legal,
or cost constraints \cite{rubin2008objective,angrist2008mostly}. An attractive alternative that has been used in
statistics, economics, and social sciences simulates control
experiments using available {\em observational data}.  While we can no longer assume that
the treatment has been randomly assigned, under appropriate
assumptions we can still \cutr{learn causal relationships.}
\revm{estimate causal effects. Rubin’s Potential Outcome Framework \cite{rubin1970thesis} and Pearl’s Causal Models \cite{PearlBook2000} (reviewed in Section~\ref{sec:background}) are two well-established frameworks which have been extensively studied in the literature and used in various applications for causal inference from observational data \cite{banerjee2011poor,rubin2006matched,rudin,ogburn2017causal,angrist2008mostly}.} A quick search on SemanticScholar reveals a growing interest in observational
studies compared to controlled experiments, as shown in  Figure~\ref{fig:obstrend}.

\cutr{
\revm{Two 
established frameworks  for observational causal analysis, 
Rubin's Potential
Outcome Framework~\cite{rubin1970thesis} and Pearl's Causal
Models~\cite{PearlBook2000} (reviewed in Section~\ref{sec:background}),
have been  extensively studied 
in the literature of statistics, artificial intelligence,
social science, health, 
and other application domains.}
}

 \revm{Causal frameworks, however, rely on the critical assumption that
the units of study are sampled from a population of homogeneous units;
in other words, the data can be represented in a single flat table. This assumption is called the unit homogeneity assumption \cite{Holland1986,angrist2008mostly}.} In many real-world
settings, however, the study domain consists of {\em heterogeneous units} that
have a {\em complex relational structure}; and the data is naturally represented as multiple related tables. \revm{For instance, as presented later in our experiments with real data \cite{mimic, healthcare_cost_and_utilization_project_hcup_hcup_2006}, 
hospitals can record   in several tables information about patients, medical practitioners, hospital stays, treatments performed, insurance, bills, and so on. }  Standard
notions used in causal analysis --- such as \revm{units or subjects who receive a treatment in causal analysis} ---
no longer \revm{readily} apply to relational data, prohibiting us from adopting existing causal inference frameworks to relational domains.
We illustrate these challenges with the following example.

\begin{figure}
    \centering
    \includegraphics[scale=0.3]{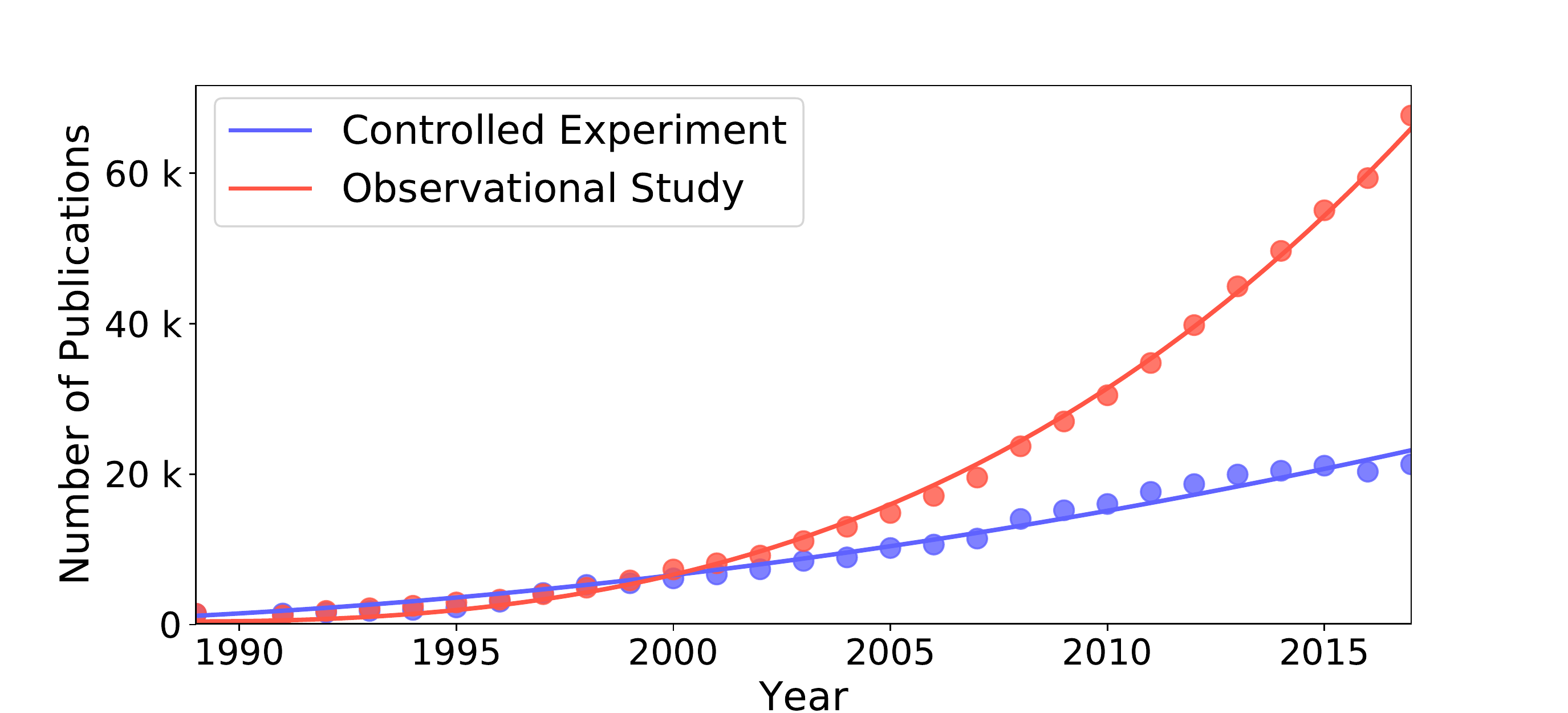}
    \caption{Number of publications that include observational studies
      vs controlled experiments (obtained from
      SemanticScholar~\cite{SemanticSchoolar}).
      }
    \label{fig:obstrend}
\end{figure}

\begin{example} (\data) \em \label{ex:univdomain:causalmodel}
  OpenReview~\cite{OpenReview} is a collection of paper submissions from several conferences, mostly in ML and AI, along with their reviews. Significantly, the collection contains review scores and author information for both
  accepted and rejected papers.  Scopus~\cite{scopus} is a large, well
  maintained database of peer-reviewed literature, including
  scientific journals, books and conference proceedings. The Shanghai University Ranking~\cite{shanghai} is one of three authoritative global university rankings.  We crawled
  and integrated these sources to produce a relational database,
  which we show in simplified form in
  Figure~\ref{fig:example_instance}.  Data sources like this represent a
  treasure trove of information for the leadership of scientific
  conferences and journals.  For example, they can help answer
  questions like ``Does double-blind reviewing achieve its desired effect?'' or
  ``Does increasing (or decreasing) the page limit affect paper quality, and if so, how?''.  To answer these real-life questions,
  discovering association is not sufficient; instead, decision
  makers need to know if there exists a causal effect. For
  example, suppose a conference is currently requiring double-blind
  submissions, and the leadership is questioning the effectiveness of this policy.  If leadership reverts to single-blind submissions, would that represent an
  unfair advantage for authors from prestigious institutions? 
  Given  a dataset like that in Figure~\ref{fig:example_instance}, one can run a few
  SQL queries and check whether those authors consistently get better
  reviews at single- rather than double-blind
  venues. However, this only proves or disproves correlation, 
  not causation.  Alternatively, one could apply Rubin's Potential Outcome Framework~\cite{rubin1970thesis, rubin2005causal}, but that requires 
  data to be presented as a single table of independent units.  Doing this
  naively on our dataset (\eg, computing the universal
  table~\cite{ullm82}) means one cannot account for what statisticians refer to as {\em interference} and
  {\em contagion} effects~\cite{tchetgen2012causal, ogburn2014causal}, 
  both of which prohibit standard causal
  analysis. For instance, the prestige of an author not only influences
  his or her acceptance rate, but it also has a {\em spill-over} effect on the acceptance rate of  his or her co-authors; this is called interference.
  Further, authors' qualifications can be {\em contagious}, meaning that if a junior author collaborates frequently with a
  senior author, then the overall quality of his or her own research may
  increase over time.

\end{example}

\textbf{Our contributions.}  In this paper, we propose a declarative framework for {\em Causal
  Relational Learning}, a foundation for causal inference over relational domains. Our first contribution is a declarative language, \emph{\sys\ (Causal Relational Language)}, for representing causal background knowledge and assumptions in relational domains ({\bf Section~\ref{sec:causal-framework}}). 
\sys\ can represent complex causal models  using just a few rules. The syntax of \sys~ is designed to be intuitive for users to represent complex causal models and ask causal queries, while the details of their semantics and query answering are abstracted from the users who need not be statisticians. 

Our second contribution is to define semantics for \textit{complex causal queries} where the treatment units and outcome units might heterogeneous and controlling for confounding may require performing multiple joins and aggregates ({\bf Section~\ref{sec:semantics}}). Using \sys, we can answer complex causal queries such as:``what is the effect of not having an insurance on mortality of a patient in a critical care unit?'', where we are interested in estimating the \textit{average treatment effect} (defined later), or ``what is the effect of authors' collaborators' prestige on acceptance of a paper?'', where we are interested in estimating the \textit{average relational effect}; several other types of queries are also supported.

Our third contribution consists of an algorithm for answering causal queries from the given relational data ({\bf Section~\ref{sec:causal-queries:answering}}). The algorithm performs a static analysis of the causal query, and it constructs a unit-table specific to the query and the relational causal model by identifying a set of attributes that are sufficient for confounding adjustment. The constructed unit-table is amenable to sound causal inference using existing techniques.  


Finally, we present an end-to-end experimental evaluation of \sys\ on both real and synthetic data ({\bf Section~\ref{sec:experiments}}). The experiments conducted on the following real-world relational datasets: 1) \data~ \cite{OpenReview,shanghai,rose_pybliometrics:_2019}, 2) \mimicdata~ (Medical Information Mart for Intensive Care Data) \cite{mimic}, and 3) \nisdata~ (National Inpatient Sample Data) \cite{healthcare_cost_and_utilization_project_hcup_hcup_2006}. We examine the following causal queries:
\begin{itemize}
    \item \data. What is the effect of authors' prestige on the scores given by the receivers under single-blind and double-blind review processes?
    \item \mimicdata. What is the effect of not having insurance on patient's mortality and length of hospital stay? 
    \item \nisdata. What is the effect of hospital size on healthcare affordability?
\end{itemize}
In each setting,
we report contrasts between correlation and causation, further highlighting the need for principled causal analysis. Evaluation of \sys\ on synthetic data showed that 
causal analysis ignoring the relational structure of data failed to recover the ground truth, but \sys~ successfully recovered accurate results. 
  \begin{figure}
    \centering
    \input{instance.tex}

    \caption{{A multi-relational 
     \data\ instance.}}
     \vspace{-0.5cm}
    \label{fig:example_instance}
\end{figure}

%% file: instance.tex

{\footnotesize

\begin{tabular}{ccccc}
\begin{tabular}[t]{|c|c|c|}
\hline
\multicolumn{3}{|c|}{\textbf{Authors}} \\ \hline
person & prestige & qualification \\
& & (h-index)\\\hline
 Bob & 1 & 50 \\
 Carlos & 0 & 20 \\
 Eva  & 1 & 2\\ \hline
\end{tabular}
\hspace{3mm}
\begin{tabular}[t]{|c|c|}
\hline
\multicolumn{2}{|c|}{\textbf{Submissions}} \\ \hline
sub & score \\\hline
s1 & 0.75 \\
s2 & 0.4 \\
s3 & 0.1 \\ \hline
\end{tabular}
\hspace{3mm}
\begin{tabular}[t]{|c|c|}
\hline
\multicolumn{2}{|c|}{\textbf{Authorship}}\\
\hline
person & sub \\\hline
Bob & s1 \\
Eva & s1 \\
Eva & s2 \\
Eva & s3 \\
Carlos & s3 \\ \hline
\end{tabular}
\hspace{3mm}
\begin{tabular}[t]{|c|c|}
\hline
\multicolumn{2}{|c|}{\textbf{Submitted}}\\\hline
sub & conf \\\hline
s1 & ConfDB \\
s2 & ConfAI \\
s3 & ConfAI \\ \hline
\end{tabular}
\hspace{3mm}
\begin{tabular}[t]{|c|c|}
\hline
\multicolumn{2}{|c|}{\textbf{Conferences}}\\
\hline
conf & blind \\\hline
ConfDB & Single \\
ConfAI & Double \\ \hline
\end{tabular}
\end{tabular}
}

%% file: Background2.tex
\vspace{-0.3cm}
\section{Background on Causal Analysis}\label{sec:background}
%


In this section we \revm{review fundamental concepts in causal analysis.} \revm{We use capital letters $X$ 
to denote \revm{random} variables, and use lower case letters $x$
to denote their values.  We use boldface $\mb X$, $\mb x$ 
to denote tuples of random variables and constants respectively; and $Dom(X)$ denotes the domain of variable $X$. }


 \paragraph*{\bf Probabilistic causal models.}
A probabilistic causal model \cite{pearl2009causality} is a tuple
$\cm = \langle \mb U, \mb V , \pr_{\mb U}, \mb F \rangle$, where \revm{
$\mb U$ is a set of \emph{exogenous} variables that cannot be observed,   
$\mb V$ is a set of \emph{observable or endogenous} variables, 
and $\pr_{\mb U}$
is a joint probability distribution on the exogenous variables
$\mb U$. The set $\mb F = (F_X)_{X \in \mb V}$ 
is a set of \emph{non-parametric  structural equations}
of the form 
$F_X : Dom(\Pa_{\mb V}(X)) \times Dom(\Pa_{\mb U}(X)) \rightarrow
Dom(X)$, where  $\Pa_{\mb U}(X) \subseteq \mb{U}$ and
$\Pa_{\mb V}(X) \subseteq \mb V-\set{X}$ are called the \textit{exogenous
parents} and \textit{endogenous parents} of $X$ respectively.}
Intuitively, the exogenous variables $\mb U$ are not known,
but we know their probability distribution; the endogenous
variables are completely determined by their parents (exogenous and/or
endogenous). 
\cutr{
\sout{In this paper, given users' background knowledge and assumptions about causal models, we use observational data to validate assumptions and quantify causal effects. Background knowledge is {\em fundamentally} required for causal inference \cutr{\cite{pearl2009causal, rubin2005causal}}
}
}

\paragraph*{\bf Causal DAGs.}
A probabilistic causal model is
associated with a {\em causal DAG} (directed acyclic graph) 
$\cg$, whose nodes are the
endogenous variables $\mb V$, and whose edges are all pairs $(Z,X)$
such that $Z \in \Pa_{\mb V}(X)$. \cutr{We write $Z \rightarrow X$ for an
edge;} 
The causal DAG hides
exogenous variables (since we cannot observe them anyway) and instead
captures their effect by defining a probability distribution
\revm{$\pr_{\mb U}$} on the endogenous variables.\footnote{ This is possible under the
{\em causal sufficiency} assumption:
for any two variables $X, Y \in \mb{V}$, their exogenous
  parents are disjoint and independent
  $\Pa_{\mb U}(X) \indep \Pa_{\mb U}(Y)$.  When this assumption fails,
  one adds more endogenous variables to the model to expose their
  dependencies.}
\revm{We will only refer to endogenous variables in the rest of the paper
and drop the subscript $\mb{V}$ from $\Pa_{\mb{V}}$. 
Similarly,  we will  drop the subscript $\mb U$ from the probability distribution $\pr_{\mb U}$ when it is clear from the context.}
Then the formula for $\pr(\mb V)$ is the same as that
for a Bayesian network:

%
{
\begin{align}
\pr(\mb V) = & \prod_{X \in \mb V} \pr(X | \Pa(X))  \label{eq:bayesian}
\end{align}
}
\begin{figure}[t]
    \centering
    \includegraphics[scale=0.3]{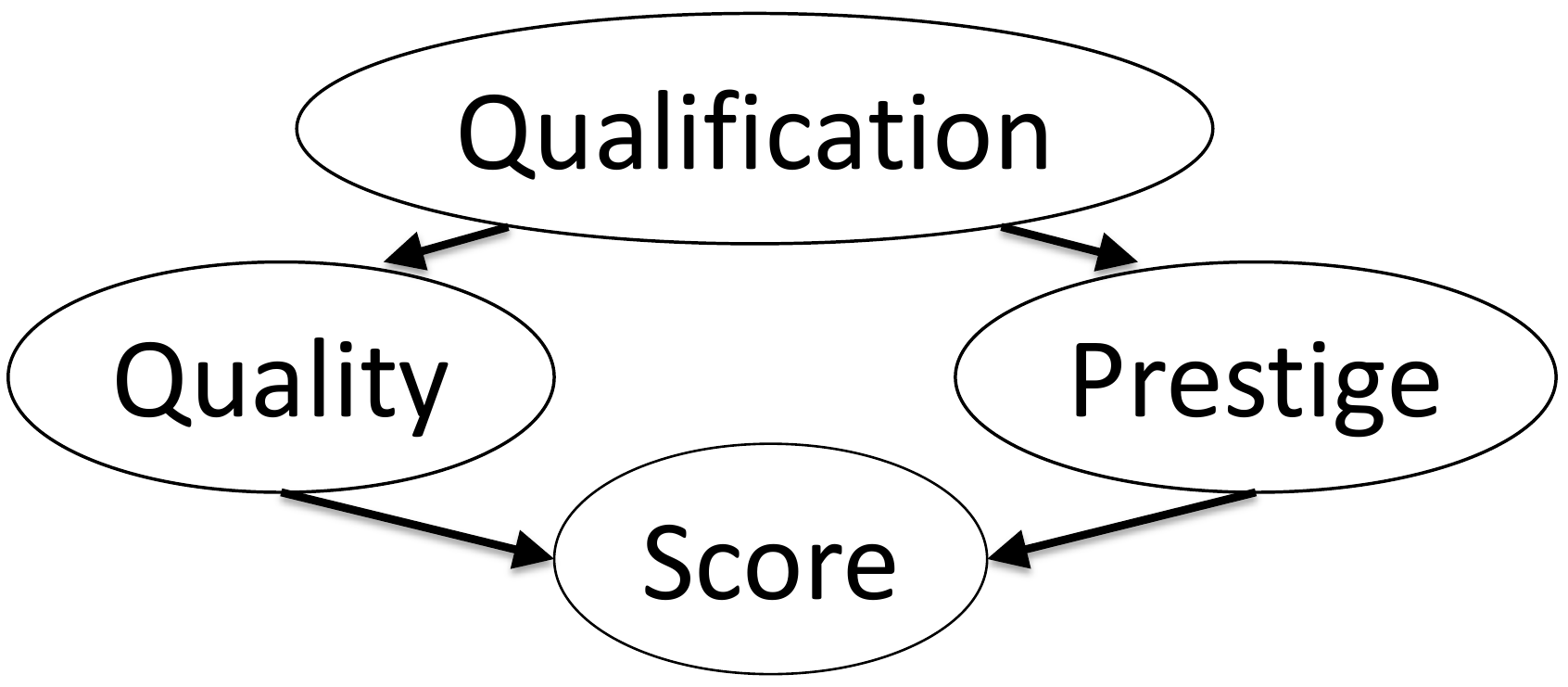}
    \caption{A standard 
    causal DAG for Example~\ref{eg:instance}. 
    }
    \label{fig:simple-causal-dag}
\end{figure}

\par
Figure~\ref{fig:simple-causal-dag} shows a simple
example of a causal graph \revm{based  on Example~\ref{ex:univdomain:causalmodel}: the Score of a paper is affected by its
Quality and by the Prestige of the author (assuming the reviews are
single blind), whereas both Quality and Prestige are affected by the author's
Qualification. Here $\mb V$ = \{\texttt{Qualification}, \texttt{Quality}, \texttt{Prestige}, \texttt{Score}\} are endogenous variables, $\mb U$ endogenous variables are unknown (e.g., mood of a reviewer while reviewing the paper, the expected number of papers to be accepted, scores of other papers the reviewer reviewed, etc.) leading to a probability distribution on $\mb V$. The  dependencies can be represented by three structural  equations:
{
\begin{align}\label{eq:se_background}
\text{Quality} &\sem   \text{Qualification}; & \ \  
\text{Prestige} &\sem  \text{Qualification};\nonumber\\
\text{Score} &\sem  \text{Quality}, \text{Prestige}.
\end{align}
}
\vspace*{-3mm}
}

\cut{In this paper we will assume that the causal DAG is
known,\footnote{There exists an extensive literature on learning the
  causal DAG from data.} and one uses some observational data in
order to learn the conditional probability distribution
(\ref{eq:bayesian}).}  


\vspace*{-3mm}
\paragraph*{\bf Interventions and the \texttt{do} operator.} \revm{Causal models give semantics to {\em interventions}.} An intervention represents actively setting an endogenous variable to some fixed value and observing the effect denoted by the $do$-operator introduced by 
Pearl~\cite{pearl2009causal}.
Formally, an intervention \revm{$\Do(\mb W = \mb w)$} consists of setting
variables \revm{$\mb W \subseteq \mb V$} to some values \revm{$\mb W= \mb w$}, and it defines the
probability distribution \revm{$\pr(\mb V | \Do(\mb W = \mb w))$} given by
(\ref{eq:bayesian}), where we remove all factors $\pr(X | \Pa(X))$,
where $X \in \mb W$.  In other words, we modify the causal DAG by
removing all edges entering the variables $\mb W$ on which we intervene; this fundamentally differs from conditioning,
\revm{$\pr(\mb V | \mb W = \mb w)$}.  Pearl has an extensive discussion of
the rationale for the $do$-operator and describes several equivalent
formulas for estimating the effect of \revm{$\Do(\mb W = \mb w)$} from \revm{an} observed distribution \cite{pearl2009causality}.

%
%
%
%
\vspace{-0.2cm}
\paragraph*{\bf Average treatment effect (ATE).}
The causal analysis estimates the effect treatment variable $T$ (typically a binary variable) on some outcome variable $Y$. This effect is often measured by the following quantity known
as the {\em average treatment effect (ATE)}, which is
expressed as follows in our notation:
\begin{align}
\ate(Y, T)& = \Ex[Y|do( T= 1)]-\Ex[ Y|do( T= 0)] 
\label{eq:ate}
\end{align}
Much of the literature on causal inference in statistics addresses efficient estimation of $\ate$ from observational data.

\revm{
\paragraph*{\bf Unit of analysis and SUTVA.} Both Pearl's \cite{pearl2009causality} and Rubin's causality \cite{rubin2005causal} rely on the assumption that the study domain consists of a set of \emph{units}, or physical objects (e.g., authors, patients, publications, etc.) that can be subject to a treatment/intervention and exhibit a response to it. Furthermore, they rely on the assumption of   {\em no interference between the units} or \emph{Stable Unit Treatment Value Assumption (SUTVA)} \cite{rubin2005causal}. Intuitively SUTVA states that intervening on or treating a unit does not have any consequences on the response of other units. In settings where the units of analysis are relationally connected, this assumption is typically violated. In Example~\ref{ex:univdomain:causalmodel}, prestige of an author (treatment) influences the acceptance chance (response) of his or her co-author(s) and collaborator(s) which leads to the violation of SUTVA.
}

\paragraph*{\bf Related Work.}
Previous work has studied causal inference in the presence of
interference
\cite{ogburn2017causal,graham2010measuring,halloran2012causal,halloran1995causal,tchetgen2012causal,vanderweele2011bounding,aronow2017,shalizi2011homophily,ogburn2014causal,ogburn2018causal}. These works address applications such as the study of infectious diseases \cite{vanderweele2011bounding, halloran1995causal} or behavior and interactions in social networks \cite{shalizi2011homophily,ogburn2018causal, sobel2006randomized,graham2010measuring, vanderweele2011bounding, maier2013reasoning,hong2006evaluating}. But in these studies the units are still homogeneous (e.g., people connected by a social network), and they are unable to capture different \emph{entities} of interests like papers, authors, reviews and their complex many-to-many relationships that we focus on in \sys.
There has been prior work on learning causality from relational data \cite{maier2013sound,DBLP:conf/uai/2019,lee2016learning,maier2010learning}; it focuses on discovering the structure of probabilistic graphical models for this data. These models were originally proposed for Statistical Relational Learning, which aims to model a joint probability distribution over relational data amenable for probabilistic reasoning rather than causal inference ~\cite{Getoor:2007:ISR:1296231}. \revd{This line of work differs from our work in that our objective is to develop a declarative framework to answer complex causal queries about the effect of interventions, given the existing background knowledge}. 
Note that {\em causality} has been used in various contexts \cite{meliou2014causality}, namely,  to understand responsibility in query answering  \cite{MeliouGMS2011, SalimiTaPP16}, in database repair \cite{meliou2011tracing, DBLP:conf/icdt/SalimiB15,bertossi2017causes}, and to motivate explanations and diagnosis \cite{RoyS14,DBLP:conf/flairs/SalimiB16,bertossi2017causes}. It has also inspired different applications such as  hypothetical reasoning \cite{balmin2000hypothetical,lakshmanan2008if, deutch2013caravan, meliou2012tiresias}. These differ from our work in that they identify parts of the input that are correlated with the output of a transformation, which is useful but
\revm{does not reflect the true causality needed for decision making.}

%% file: Framework.tex
\section{\revm{\sys: Declarative Framework}}\label{sec:causal-framework}
In this section, we present our declarative language called \emph{CaRL (Causal Relational Language)} 
that extends causal modeling 
to relational data by allowing the user to (1) specify assumptions and background knowledge on the interactions among heterogenous units (Section \ref{sec:language}), and (2) 
\revm{pose} various causal queries (Section \ref{sec:qlanguage}).
We start with our data model, which forms the basis for our language.

\vspace{-0.2cm}
\input{Preliminaries.tex}

\vspace{-0.2cm}
\subsection{\revm{Specification of Background Knowledge by Relational Causal Rules}}
\label{sec:language}

\subsubsection{\revm{Relational causal model and rules}} \label{sec:language:rules} The first step of using \sys\ is encoding the  user's background knowledge \revm{about potential causal dependencies among attributes in an application}. 
This is expressed in  \sys\ through a set of \emph{relational causal  rules} (defined below) that capture the causal assumptions. 
\revm{We refer to the set of relational causal rules specified by the user as the \emph{relational causal model}. }
\begin{definition} \em \label{def:\rcr}
A  \emph{relational causal rule} 
over a relational causal schema $\Scm=\Scmdef$ has the following form:
{  
\begin{align}
\fAtt[\mb \bx]  \sem \fAtt_1[\mb \bx_1] , \ldots, \fAtt_k[\mb \bx_k]   \text{ WHERE } Q(\mb Y) 
\label{eq:rse}
\end{align}
}
Here, $\fAtt, \fAtt_1, \cdots, \fAtt_k \in \FAtt$ are \revm{attribute functions}, $Q$ is a \revm{(standard)} conjunctive query over the schema $\bf P$, and $\bf X$, $\bf X_i$ ($i=1,\cdots,k$), $\bf Y$ are sets of variables and/or constants.  All variables  in $\bf X\cup\bigcup_i \bf X_i$ must also occur in $\bf Y$.  We call $\fAtt[\mb \bx]$ the \emph{head} of the \rcr, $\fAtt_1[\mb \bx_1] , \ldots, \fAtt_k[\mb \bx_k]$ the \emph{body} of the \rcr, and $Q(\mb Y)$ the \emph{condition}. We denote by $\rse_\fAtt$ the set of rules with head $\fAtt$.
\end{definition}

\begin{example}\em \label{ex:univer_schema2} \revm{Consider} the following relational causal model $\rcm$ for \data\ in Figure~\ref{fig:example_instance}.
{   
\begin{align} 
\text{Prestige}[A] &\sem \text{Qualification}[A] \text{ WHERE } \text{Person}(A) \label{eq:prestige1}\\
\text{Quality}[S] &\sem \text{Qualification}[A], \text{Prestige}[A]  \text{ WHERE }   \text{Author}(A, S) \label{eq:quality1}\\
\text{Score}[S] &\sem  \text{Prestige}[A] \text{ WHERE }  \text{Author}(A, S), \label{eq:decision1} \\
\text{Score}[S] &\sem \text{Quality}[S] \text{ WHERE } \text{Submission}(S) \label{eq:decision2}
\end{align}
}
Rule (\ref{eq:prestige1}) says that the qualification of a person causally affects his or her institutions' prestige;  rule (\ref{eq:quality1}) says that the quality of a submission is affected by its authors' qualifications and prestige (authors from prestigious institutions have access to more resources);  rules (\ref{eq:decision1}) and (\ref{eq:decision2}) say that reviewers' scores are based on the quality of a submission but may also be influenced by the prestige of its authors. 
\end{example}
\reva{A major advantage of specifying background knowledge using causal rules for the users is that they simply express intuitive potential causal dependence among attributes without mentioning `how' or associating any `weight' to them\footnote{\revb{This fact, along with the declarative nature of the language, makes CaRL more friendly to users who are not causal inference experts.}}, while \sys\ uses them to answer different causal queries (Section~\ref{sec:qlanguage}).
}

\subsubsection{\revm{Grounded rules}}
A \emph{grounded rule} is a rule (\ref{eq:rse}) that contains only constants from a given instance (no variables) and has no condition (i.e., $Q \equiv \texttt{true} $).  
\revm{A relational causal rule} is a template for generating multiple grounded rules.
 
\begin{definition} \em \label{def:semantics} 
Let $\Delta$ be a relational skeleton.  Fix a rule \revm{in the form of}  (\ref{eq:rse}), and let $\mb Z$ denote all variables occurring in $\mb X \cup \mb X_1 \cup \ldots \cup \mb X_k$.  We associate to this rule the set of {\em grounded  rules}  obtained by substituting $\mb Z$ with any \revm{set of} constants $\mb z$  such that $\skl \models Q([\mb Y /  \mb z])$. In other words, the query $Q$ must be true in the database $\Delta$ after substituting the variables $\mb Z$ with the constants $\mb z$ and treating the variables $\mb Y-\mb Z$ as existentially quantified.  
\end{definition}


\subsubsection{\revm{Relational causal graphs}} 
Given a relational causal model $\rcm$ \revm{comprising a set of relational causal rules} and a relational skeleton $\Delta$ \revm{comprising the entities and relationships in an instance,}
\cutr{, we associate both with the set of all grounded rules, denoted 
}
\revm{$\rcm^\skl$ denotes the set of all grounded rules.}  
\cutr{generated by the rules in $\rcm$.
}
From $\rcm^\skl$, we construct the {\em relational causal graph} $\gcg$. \revm{The vertices of $\gcg$} (denoted $\FAtt^{\skl}$) comprise all grounded attributes $\fAtt[\mb x]$ in $\rcm^{\skl}$ \revm{denoted $\fAtt^{\skl}$ -- recall that $\mb x$ represents a tuple of constants, an attribute function $\fAtt$ corresponding to an entity has a single constant parameter as in Example~\ref{eg:instance}, but $\fAtt$ corresponding to a relationship predicate will have multiple parameters. The edges of $\gcg$} are all pairs $(\fAtt[\mb x],\fAtt_j[\mb x_j])$ where $\fAtt[\mb x]$ and $\fAtt_j[\mb x_j]$ appear in the head and body respectively of a grounded rule (\ref{eq:rse}). \revm{We assume that the relational causal model is non-recursive, therefore, the causal graph is a DAG}\footnote{\revm{While our language allows for recursive rules which capture feed-back loops and contagion, their treatment is beyond the scope of the paper and is an interesting future work.}}. 
\begin{figure}
  \centering
  \includegraphics[scale=0.55]{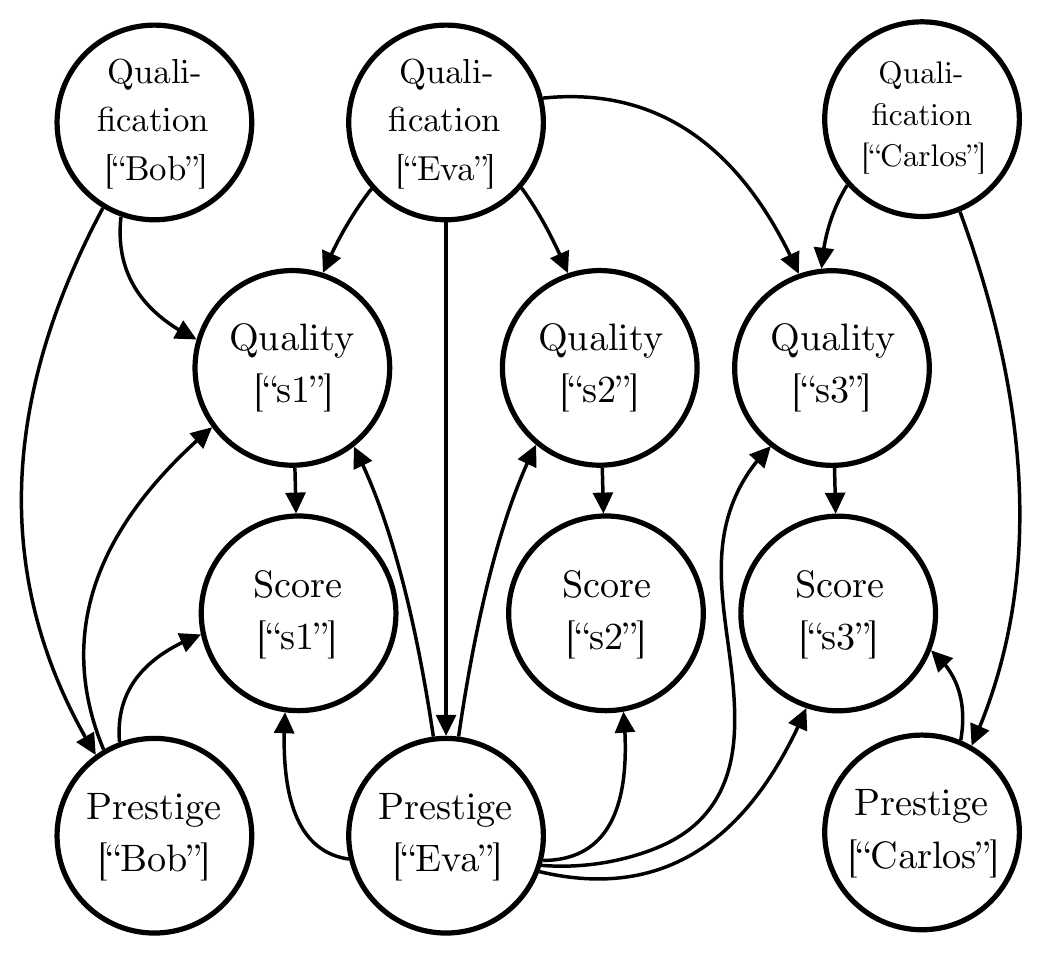}
  \caption{\revm{Relational causal graph} corresponding to the grounded rules in Example~\ref{ex:univer_schema2}.}
  \vspace{-2mm}
  \label{fig:GCD}
\end{figure}

\begin{example} \em
  Given the skeleton $\Delta$ in Figure~\ref{fig:example_instance}, $\rcm$ generates the following grounded rules:
{   
\begin{align} 
\text{Prestige}[``Bob"] &\sem  \text{Qualification}[``Bob"] ~~~~\text{-- \em(also for ``Carlos'', ``Eva'')}\nonumber\\
\text{Quality}[``s1"] &\sem \text{Qualification}[``Bob"],   \text{Qualification}[``Eva"]\nonumber\\
\text{Quality}[``s2"] &\sem \text{Qualification}[``Eva"]\nonumber\\
\text{Quality}[``s3"] &\sem \text{Qualification}[``Carlos"],  \text{Qualification}[``Eva"]\nonumber\\
\text{Score}[``s1"] &\sem \text{Quality}[``s1"], \text{Prestige}[``Bob"],    \text{Prestige}[``Eva"] 
\nonumber\\
\text{Score}[``s2"] &\sem \text{Quality}[``s2"], \text{Prestige}[``Eva"]\nonumber\\
\text{Score}[``s3"] &\sem \text{Quality}[``s3"], \text{Prestige}[``Carlos"],  \text{Prestige}[``Eva"]\label{eq:decision_s2}
\end{align}
}
These in turn lead to the causal graph shown in Figure \ref{fig:GCD}. 
\end{example}
\revm{Note that the relational causal graph in Figure~\ref{fig:GCD} is an extension of the  standard causal DAG (by Pearl's model \cite{PearlBook2000}) shown in Figure~\ref{fig:simple-causal-dag}: the latter describes the potential causal dependence of the attributes whereas the former describes a more fine grained version based on the entities and relationships in the relational data. For example, we do not have a single node $\mathit{Score}$, as in Figure~\ref{fig:simple-causal-dag}, but instead have many nodes $\mathit{Score}["s1"]$, $\mathit{Score}["s2"]$, etc. one for each submission in $\Delta$ in Figure~\ref{fig:GCD}.} As in Section~\ref{sec:background}, the causal graph $\gcg$ defines a joint probability distribution
\begin{equation}
\pr\big(\fAtt[\mb x] \mid \Pa(\fAtt[\mb x]\big) \label{eq:condl_prob}
\end{equation}
with one conditional probability for each grounded attribute $\fAtt[\mb x]$;  we describe these conditional probabilities in Section~\ref{sec:prob-rel-dags}.

\revm{
\subsubsection{Aggregated rules}\label{sec:agg_rules}
Using \sys, one can 
extend the set of attribute functions $\FAtt$ with
new aggregated attribute functions 
using one of the {\em aggregate  rules} of the following forms. For $\fAtt \in \FAtt$ 
{  
\begin{align}
\AGG\_\fAtt[\mb W] &\sem \fAtt[\mb X] \ \text{WHERE} \ Q(\mb Z) 
\end{align} 
}
Here, $\mb Z \supseteq \mb X \cup \mb W$ and $\AGG$ is an aggregate function on $\fAtt$, e.g., AVG (average) and  VAR (variance). The new aggregated attribute functions $AGG\_\fAtt$  are included in the extended attribute functions $\FAtt$ (for simplicity, we use $\FAtt$ for both given and extended attribute functions). 
\cutr{The semantics of aggregated  rules are defined similarly to rules in Section~\ref{sec:causal-framework}, \ie,
}
Similar to relational causal rules, aggregated rules define a set of grounded  rules with corresponding vertices and edges in the relational causal graph $\gcg$. However, instead of a conditional probability distribution, a deterministic  function $\AGG(\Pa({\AGG\_\outc}[\mb w]))$ will be associated with each $\AGG\_\outc[\mb w] \in \AGG\_\outc^{\skl}$.  For example, the following aggregate rule defines the average review score for each author.
{  
\begin{align}
 \text{AVG\_Score}[A]&\sem \text{Score}[S]  \ \text{WHERE} \  \text{Author}(A,S)
\label{eq:evg_grade:course}
\end{align} 
}
\noindent Figure~\ref{fig:agg-GCD} shows the extension of Figure~\ref{fig:GCD} with (\ref{eq:evg_grade:course}). 

 {\scriptsize
\begin{figure}[t]
    \centering
    \includegraphics[scale=0.6]{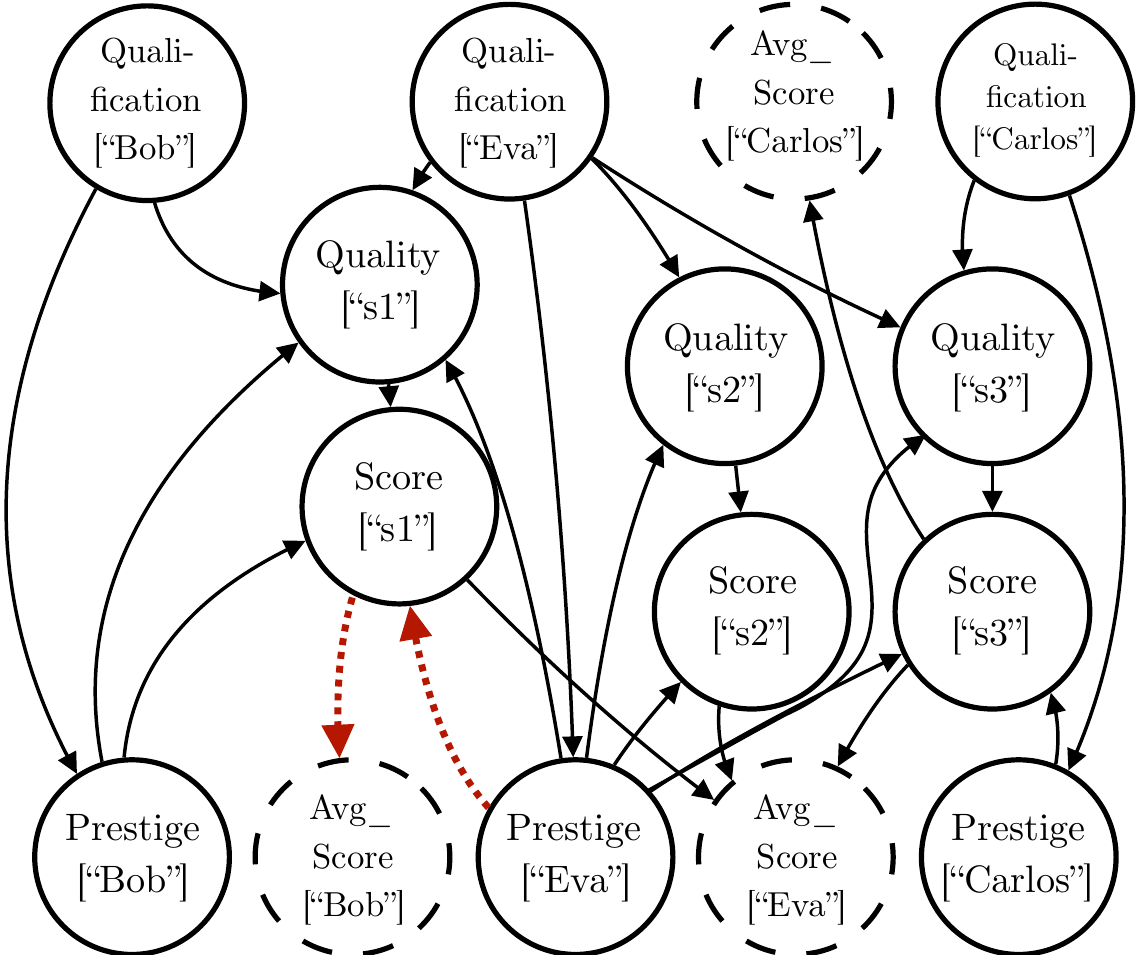}
    \caption{Extended relational causal graph from Figure~\ref{fig:GCD} with aggregated attribute $\text{AVG\_Score}[A]$ \revm{by (\ref{eq:evg_grade:course}). The directed path from relational peer Eva's prestige to average score of Bob is highlighted (Section~\ref{sec:rel-paths-peers}).}}
    \label{fig:agg-GCD}
      \vspace{-2mm}
\end{figure}
}

}

\subsection{\revm{Causal Query Language in \sys}} 
\label{sec:qlanguage}

Once the relational causal model $\rcm$ is \revm{specified}, users can start asking causal queries.  \sys\ supports three types of causal queries of the following form \revm{(their semantics are discussed in Section~\ref{sec:causal-queries} and answering these queries is discussed in Section~\ref{sec:causal-queries:answering})}.
\revm{The \emph{ATE query} extends the notion standard ATE (discussed in Section~\ref{sec:background}) for relational data.  \sys\ also supports queries for aggregated response,
isolated effect and relational effect.}

\cut{
\sout{We describe here their syntax and present their semantics only informally; the formal semantics is based on a non-trivial extension of the $do$-operator that is explained in Section~\ref{sec:causal-query}.
}
}
%
\begin{sloppypar}
  \paragraph{\bf  Average treatment effect (ATE) query.} \revd{An ATE query estimates the average treatment effect (see Section~\ref{sec:background}) of a \emph{treatment attribute} $\treat[\mb X]$ 
  on a response attribute $\outc[\mb X']$ and has the following form:} \revm{(formally defined in Section~\ref{sec:ate})}
\begin{align}
\outc[\mb X'] \sem \treat[\mb X]?\label{eq:simple_ATE}
\end{align} 
This asks ``what is the effect of $\treat$ on $\outc$?''.  
For example, the query $\mathit{Score}[\revm{S}] \sem \mathit{Prestige}[A]?$ computes the ATE of $\mathit{Prestige}$ \revm{of authors} on $\mathit{Score}$ \revm{of a paper}, \ie, it compares papers’ scores in two hypothetical worlds in which all authors are and are not affiliated with prestigious institutions \revm{(the causal effects of `some' authors being from prestigious institutions can be estimated from the relational effects queries described below)}.  
\revd{Following the standard assumption of binary treatments in the causality literature, we require the treatment attribute to be of binary domain, which can be enforced by using a threshold or a predicate on a non-binary domain.}
\end{sloppypar}

\paragraph{\bf Aggregated response query.} \cutr{Alternatively, one can ask for the}
\revm{An aggregated response query allows causal analysis on an aggregated form of the response variable and has the following syntax (formally defined in Section~\ref{sec:agg-response-queries})}:
\begin{align}
\text{\revm{AGG}}\_\outc[\mb X'] \sem \treat[\mb X]?\label{eq:agg_ATE}
\end{align} 
\revm{For example, $\text{AVG\_Score}[\revm{S}] \sem \text{Prestige}[A]?$ computes the treatment effect of the prestige of authors on the average score received by an author.}

\cutr{
single value, namely, the average of all ATEs for all submissions.

\red{SR: DON'T FOLLOW THIS -- IS NOT ATE BY DEFAULT AN AGGREGATE/EXPECTATION? IF IT IS INDIVIDUAL IT IS CATE OR CONDITIONAL ATE. WE SHOULD USE CATE FOR NO. 1 AND ATE FOR THIS ONE NO. 2.}

\babak{ATE is an expectation over all units. Aggregated outcomes define change the units of analysis. These are different.}
}

\paragraph{\bf \revm{Relational, isolated, and overall effects queries}.} In relational domains, units that are relationally connected can have a causal influence on each other. For example, the Prestige of an author not only influences their average submission scores but also their collaborators' average submission scores. To measure such complex relational causal interactions, \revm{\sys\ supports 
queries of the following form that output three quantities: relational, isolated, and overall causal effects
(formally defined in Section~\ref{sec:relational-isolated})}:
{   
\begin{align}
\outc[\mb X'] \sem \treat[\mb X] \ ? \  \mathtt{WHEN} \  \langle \mathit{cnd} \rangle \  \ \mathtt{PEERS} \ \mathtt{TREATED} 
\label{eq:relationalquery}
\end{align}
}
where $\langle \mathit{cnd} \rangle$ is a condition with the following grammar:
{   
\begin{align}
\langle \mathit{cnd} \rangle  \leftarrow  & \langle  \mathtt{LESS}\mid \mathtt{MORE} \rangle \ \mathtt{THAN} \ k\% \mid   \mathtt{AT \ \langle  \mathtt{MOST}\mid \mathtt{LEAST} \rangle \ k} \mid \nonumber \\
 &  \mathtt{EXACTLY \ k} \   \mid \mathtt{ALL} \ \mid \mathtt{NONE} \label{eq:query_grammar}
\end{align} 
}
For example, the query
{   
$\text{Score}[S] \sem \text{Prestige}[A]? \ \mathtt{WHEN} \  \mathtt{ALL} \ \mathtt{PEERS} \ \mathtt{TREATED}$ 
}
\revd{ \noindent computes three values for (i) isolated (an author's prestige), (ii) relational (his/her coauthor's prestige), and (iii) overall (all authors' prestige) effect of prestige on a submission's score.}


\cutr{
\red{SR: THE FOLLOWING SHOULD BE MOVED TO SECTION 4}
\revd{For all the treatment unit, consider the vector $\mathbf{T}[X]$ where each entry $T[X_i]$ is the treatment indicator for the corresponding treatment unit $\mb x_i$.
The isolated effect
measures the extent to which the review scores received by an author are influenced by his or her own prestige $\mathbb{E}[Y_i[X]|do(T[X_i]=1),\mathbf{T}[X]\setminus T_i[X]]-\mathbb{E}[Y_i[X]|do(T_i[X]=0),\mathbf{T}[X]\setminus T_i[X]]$. The relational effect measures the extent to which the review scores of an author are influenced
by the prestige of his or her collaborators $\mathbb{E}[Y_i[X]|do(\mathbf{T}[X]\setminus T_i[X]=\mathbf{1}), T_i[X]]-\mathbb{E}[Y_i[X]|do(\mathbf{T}[X]\setminus T_i[X] = \mathbf{0}),T_i[X]]$.}
}

%% file: Preliminaries.tex
\subsection{\revm{Data Model: Schema and Instance}}\label{sec:prelim}

\par
\textbf{\revm{Relational causal schema (schema).}~}  
\revm{The input schema for \sys\ corresponds to any standard multi-relational database, e.g., \data\ in Figure~\ref{fig:example_instance},
but we assume the data is given in the following `entity-relationship-attribute' form for a simpler generalization of Pearl's causal models.}
A {\em relational causal schema} is a tuple $\Scm=\Scmdef$, 
\revm{
where $\Pred=\Ent \cup \Rel$ 
\cutr{
$=\{\pred_1, \ldots, \pred_{n}\}$
\sout{is a standard relational schema that }
}
represents a set of \emph{entities} $\Ent$ and their \emph{relationships} $\Rel$, and 
$\FAtt$ \cutr{=$\{\fAtt_1, \ldots, \fAtt_{k}\}$}
represents a set of {\em attribute functions} \revd{(or simply {\em attributes})}
\cutr{
\sout{with fixed arity, domain and range $Range(\fAtt)$}
}
that encode the \cutr{descriptive values }
standard attributes of the entities and their relationships, with the only difference that some of these attributes may be `\emph{unobserved}' with missing values in all instances.} 
\revm{The entities and their relationships are denoted by $\pred(.)$, the attribute functions are denoted by  $\fAtt[.]$, and $\FAttrobs \subseteq \FAtt$ denotes the set of observed attribute functions.
We illustrate the mapping from standard relational model to relational causal schema using our running example\footnote{\revm{For the purpose of this paper, we assume that the input is given in the form of relational causal schema, whereas we envision that in an end-to-end system with a graphical user interface, this mapping will be semi-automatically done with user's input.}}.
}
\vspace*{-1mm}
\begin{example} \em \label{ex:univer_schema} 
The relational causal schema corresponding to the \revm{relational \data} in Figure~\ref{fig:example_instance} (with renames) is:
%
{   
\begin{align*}
\Pred& =\text{Person(A)}, \text{Author}(A, S), \text{Submission}(S),   \text{Submitted}(S, C),  \text{Conference}(C)\\
\FAtt& =\text{Prestige}[A], \text{Qualification}[A], \text{Score}[S],  \text{Blind}[C], \text{Quality}[S]
\end{align*}
} 
\cutr{\noindent The first row shows the relational schema. These consist}
\revm{Here $\Pred$ consists of
entities in the \data: $\Ent = \{\text{People}$, $\text{Submission}$,
$\text{Conference}\}$ and their relationships $\Rel = \{\text{Authors}$,
$\text{Submitted}\}$;}  \revm{The attribute function $\FAtt$
corresponds to the attributes of these entities and relationships:  $\text{Prestige}[A]$ = the prestige
of the author's institution (e.g., rankings);
$\text{Qualification}[A]$ 
= the qualification of an author 
by h-index\footnote{\revm{There can be other measures of qualifications as well, e.g., the number of publications or citations, or the experience in terms of years.}});
$\text{Score}[S] \in [0,1]$ = the average
score reviewers gave to a submission;  $\text{Blind}[C]$ =
whether a conference review policy is single or double blind; 
 $\text{Quality}[S]$ = the quality of a submission. Note that $\text{Quality}$ in $\FAtt$ is missing in the Submissions table 
 in Figure~\ref{fig:example_instance}, since it
is an unobserved attribute function, and will be used in causal  analysis based on our background knowledge that quality of a submission may have an impact on its score. 
}
\end{example}

\par
\textbf{\revm{Observed instance and relational skeleton (instance).}}  
\cutr{Our language represents the causal relationships between attribute functions given a set of entities and their relationships. Therefore, 
}
\revm{Similar to a standard database instance given a standard relational schema (as shown in Figure~\ref{fig:example_instance}), an \emph{observed relational instance} (or simply an \emph{instance}) conforms to a given relational causal schema $\Scm=\Scmdef$ with specific values (i.e., constants), however some (unobserved) attribute functions may be missing in the instance (like `Quality'). The set of (constant or grounded) entities and relationships in an instance (excluding the grounded attribute functions) is referred  to as the \emph{relational skeleton} of the instance and denoted by $\Delta$.}
\begin{example} \em \label{eg:instance}
\revm{For relational causal schema given in Example~\ref{ex:univer_schema} and the instance in Figure~\ref{fig:example_instance}, the relational skeleton comprises entities and relationships like {\tt Person(``Bob''), 
Submission(``s1''), 
Author(``Bob'', ``s1''), 
} 
etc. The observed instance comprises the relational skeleton and the attribute functions like {\tt Score[``s1''], Blind[``ConfDB''], 
} etc., but not unobserved attributes like {\tt Quality[``s1'']}. Note that all observed attribute functions assume a fixed value given any instance.
}
\end{example}

\cutr{
we refer to a  standard relational database instance from schema $\pred$ 
as a {\em relational skeleton}. A relational skeleton $\skl$, together with the values of a subset of observed attribute functions $\FAtt$, is called an {\em observed relational instance}.

Figure~\ref{fig:example_instance} shows a relational skeleton of the
observed instance. }

%% file: semantics.tex
\section{\revm{Semantics For Relational Causal Analysis}}
\label{sec:semantics}
This section defines semantics of the causal queries described in Section~\ref{sec:qlanguage}. \revm{We fix} a relational causal schema $\Scm $, a relational skeleton $\skl$, and a relational causal model  $\rcm$ with a corresponding grounded causal graph $\gcg$. For an attribute function $\fAtt \in \FAtt$, denote $\xunit$ to be the set of all tuples of grounded entities $\mb x$ such that
$\fAtt[\mb x] \in \FAtt^{\skl}$. For example, $\unit_\text{Prestige}$ consists of all authors, e.g.,  \{``Bob", ``Eva", ``Carlos"\}, whereas
$\unit_{\text{Score}}$ consists of all submissions, e.g., \{\emph{``s1"}, \emph{``s2"}\}.  We refer to each
element $\mb x \in \xunit$ as a {\em unit} of an attribute function $\fAtt$.  


\subsection{\revm{Probability Distribution for \sys 
}}\label{sec:prob-rel-dags}
\input{probability-model.tex}
\vspace*{-6mm}
\subsection{Treated and Response Units} \label{sec:treated-response}
\revm{In standard causal analysis, the units can be considered tuples in a single unit table, with one attribute corresponding to the treatment and another attribute corresponding to the response. For instance, in the schema given in Figure~\ref{fig:example_instance} and relational causal graph in Figure~\ref{fig:GCD}, one could analyze the causal effect of qualification of 
authors on their prestige,
and then the `authors' form both the treated and response units. 
In contrast, for multi-relational causal analysis in \sys, when one analyzes the causal effect of prestige of authors on scores of submissions, then intuitively the authors form the treated units and the submissions form the response units. 
Even when authors (or submissions) form both the treated and response units, \sys\ allows inclusion of additional attributes from other relations that are covariates and required for answering causal queries (see Section~\ref{sec:cov-detec}).
Next we formally define these concepts.} 
\par
\revm{
In \revm{relational} causal analysis, we are given 
a \emph{treatment attribute function} $\treat[\mb X] \in \FAtt$ and a \emph{response attribute function} $\outc[\mb X'] \in \FAtt$; The set of \emph{units} $\tunit$ (resp. $\runit$) denotes the entities or relationships  corresponding to the treatment (resp. response) attribute function $\treat$ (resp. $\outc$).
For example, to study the effect of authors' prestige on submission scores, 
$\text{Prestige}[A]$ is the treatment attribute function and $\text{Score}[S]$ is the response attribute function, $\unit_{\text{Prestige}}$ denotes all authors as treated units and $\unit_{\text{Score}}$ denotes all submissions as response units (we assume without loss of generality that the attribute function names are unique and correspond to a single entity or relationship). 
We assume the treatment attribute has binary values whereas the response can be any real number. 
}

\cutr{
In \revm{relational} causal analysis, we are given a pair of attribute functions $\treat[\mb X]$, $\outc[\mb X'] \in \FAtt$. The goal is to estimate the effect of intervening on the \emph{\revm{treatment} attribute function} $\treat$ \revm{(e.g., $\mathit{Prestige}[A]$)} of the {\em treated units}  $\tunit$ on the \emph{\revm{response} attribute function} $\outc$ \revm{(e.g., $\mathit{Score}[S]$}) of the {\em response units} $\runit$. \sout{Hence we call $\treat$ the {\em treatment attribute function} and $\outc$ the {\em response attribute function}. For example, to study the effect of prestige on submission scores, 
$\mathit{Prestige}[A]$ is the treatment and $\mathit{Score}[S]$ is the response attribute function, the treated units are $\unit_{\mathit{Prestige}}$ (all authors), and the response units are $\unit_{\mathit{Score}}$ (all submissions). For exposition, we assume $Range(\treat)=\{0,1\}$ and $Range(\outc)=\mathbb{R}$.  
}
}

\revm{Given} a set of treated units $\tunit=\{\mb x_1, \mb x_2, \ldots \}$ and a binary vector $\trevec=(t_1, t_2, \ldots)$, we are interested in the effect of a set of interventions $\Do(\treat(\mb x_i)=t_i)$ for all treated units $x_i$, where each intervention replaces the NSE associated with $\treat(\mb x_i)$ with a constant $t_i$. In our example of the effect of prestige on score, the vector $\trevec$ corresponds to a particular assignment of prestige to all authors, e.g., the vector $\Vec{1}$ identifies an intervention that {\em hypothetically changes \revm{`all'} authors' affiliations to prestigious ones}. 
By abuse of notation, we denote with $\Do(\treat[\arbtreat]=\trevec_{\arbtreat})$ a set of interventions in which an arbitrary subset of treated units $\arbtreat \subseteq \tunit$ receive $\trevec_{\arbtreat}$ (with an implicit assumption on  the order of elements in the set $\arbtreat$).  \revm{Having treated/response units and the treatment vectors allows us to have (1) \emph{non-uniform units} that may be different entities or relationships, and (2) \emph{different types of treatments}, e.g., forcing all authors to be of prestigious institutions as $\Vec{1}$ vs. one or some of the authors from prestigious institutions as $(1, 0, 0, \cdots)$.}
\sys\ aims to answer 
\cutr{a set of} 
causal queries that compare the {\em average response} of the response units $\runit$ to two alternative intervention strategies $\trevec$ and $\trevec'$ applied to the treated units
$\tunit$, which we discuss next. 


    
    
    
 

\subsection{\revm{Relational Paths and Peers}}\label{sec:rel-paths-peers}
\revm{Before we can formalize the semantics of causal queries described in Section~\ref{sec:qlanguage}, especially for the isolated and relational effects}, we need to establish a one-to-one correspondence between treated and response units by using aggregations carefully. To this end, we first define relational paths. 

\begin{definition} \em \cutr{Given a relational causal schema $\Scm=\Scmdef$ consisting of a set for entities and relations $\Pred=\Ent \cup \Rel$, and a set of attribute functions $\FAtt$ (possibly extended with aggregates as discussed in Section~\ref{sec:agg_rules}),} \revd{ A {\em relational path} is a sequence of entities and relationships of the following form:} 
{       
\begin{align}
    \mc P: \ent_1(X_1) \xleftrightarrow{R_1(X_1,X_2)} \ent_1(X_2) \cdots  \ent_{\ell-1}(X_{\ell-1}) \xleftrightarrow{\rel_{\ell-1}(X_{\ell-1},X_{\ell})} \ent_{\ell}(X_{\ell})
    \label{eq:relationalpath}
\end{align}
}
\noindent where $\ent_i(X_i) \in \Ent$ and $\rel_{i-1}(X_{i-1},X_{i}) \in \Rel$, for $i=1, \cdots, \ell$.

\end{definition}
\noindent

For instance, ${              \text{Conference}(C)  \xleftrightarrow{\text{Submitted}(S,C)} \text{Submission}(S)}$ is a relational path in our example. The treated and response units corresponding to treatment and response attribute functions $\treat$ and $\outc$ are said to be  {\em relationally connected} if there exists a relational path $\mc P$ that includes the entities or relationships \revm{for} $\treat$ and $\outc$ \cutr{describe}
either as the endpoints in the path or as the labels of the edges at the ends of the path.  For example, for $\treat[\mb X] = \text{Prestige[A]}$ and  $\outc[\mb X'] = \text{Score[S]}$, the treatment is an attribute function of the entity $\text{Author}(A)$, the response is an attribute function of the relationship $\mathit{Author(A,S)}$, and the treated and response units are relationally connected by the following relational path:
\begin{align}
{              \text{Author}(A)  \xleftrightarrow{\text{Author}(A,S)} \text{Submission}(S)}
\label{eq:relp:tutor}
\end{align}

In this paper, we make the natural assumption that the treated and response units are relationally connected by at least one relational path \revm{as otherwise the effect of treatment on the response is not meaningful}. \cutr{, as in (\ref{eq:relationalpath}).} These units can then be unified using the aggregated response $\AGG\_\outc[\mb X]$ defined with the following aggregate  rule \revm{(see Section~\ref{sec:agg_rules})} that maps attribute $Y$ 
\revm{of response units $\runit$ to treatment units $\tunit$, where the units can be either entities or relationships}.
%
{     
\begin{align}
\AGG\_\outc[\mb X] &\sem Y[\mb X'] \ \text{WHERE} \ R_1(X_1,X_2), \ldots, R_{\ell-1}(X_{\ell-1}, X_{\ell})   \label{eq:agg_rcr_path}
\end{align} 
}
For example, to unify the treated and response units associated to $\treat[\mb X] = \text{Prestige[A]}$ and  $\outc[\mb X'] = \text{Score[S]}$, the aggregate  rule\footnote{\revm{We aggregate the response and not the treatment since aggregating treatments may lead to interventions that are not well defined.}} associated with the relational path in (\ref{eq:relp:tutor}) coincides with  (\ref{eq:evg_grade:course}): \revm{$ \text{AVG\_Score[A]} \sem \text{Score}[S]  \ \text{WHERE} \  \text{Author}(A,S)$}. 

Therefore, we assume from here on that the response units  $\runit$ are the same as the 
treated unit $\tunit$. Henceforth, we simply refer to elements of $\runit$ and $\tunit$ as {\em units} and denote them with $\uunit=\tunit=\runit$. 
\revm{In our example, after the unification, the $\text{AVG\_Score}[A]$ can be considered as a new attribute function of authors (as in a `view' in relational databases), and the authors form $\uunit$.}

\textbf{Relational Peers.} 
Next, we define the notion of relational peers of a unit, which is central to the notion of relational and isolated effects. Recall that the grounded causal graph $\gcg$ is extended with vertices and edges corresponding to aggregated attributes \revm{as discussed in Section~\ref{sec:agg_rules}}.

\begin{definition} \em Given a treated attribute function $\treat[\mb X]$, and a (possibly aggregated) response attribute function $\outc[\mb X]$, we define the {\em relational peers} of a unit $\mb x \in \uunit$ as a set of units $\nbs(\mb x) \subseteq \uunit-\{\mb x\}$ such that for each $\mb p \in \nbs(\mb x)$, there exists a directed path from 
$\treat[\mb p]$ to $\outc[\mb x]$ in $\gcg$. 
\end{definition}

For example, \revm{in Figure~\ref{fig:agg-GCD}}, treatment and aggregated response functions ${ \text{Prestige[A]}}$ and  ${ \text{AVG\_Score}[A]}$,  $\nbs(``Bob")=\{``Eva"\}$ and $\nbs(``Eva")=\{``Bob",``Carlos"\}$. In practice, the relational causal model  is expected to form relational peers  $\nbs(\mb x)$ that consist only of units that are in some {\em relational proximity} of {\mb x}, e.g., authors from the same institution, same research interests, etc.\footnote{This assumption is far less strict than than the assumption of partial interference, which is standard in statistics, to extend Rubin's causality to handle interference \cite{tchetgen2012causal}. Also note that the assumption of no interference (or SUTVA)~\cite{rubin1970thesis} translates to the statement $\nbs(\mb x) = \emptyset$ for all $\mb x \in \uunit$ in relational causal models.}

The following quantity measures the expected response of a unit $\mb x \in \uunit$ when it receives the treatment $\tre$, and its relational peers receive the vector of treatments $\nbsvec=(t_1, t_2 \ldots)$.
{      
\begin{align}
\outc_{\mb x}(\tre, \nbsvec) &\defeq  \EX[\outc[\mb x] \mid \Do\big(\treat[\mb x]=\tre\big), \Do\big(\treat[\nbs(\mb x)]=\nbsvec\big)] \label{eq:effecttemp}
\end{align} 
}
In this paper, we assume  $\Do\big(\treat[\nbs(\mb x)]=\nbsvec\big)$ is a {\em well-defined intervention} for all units $\mb x$,  \ie, it uniquely determines which relational peers of a unit would receive which treatment. For instance, this holds if $\nbs(\mb x)$ is of the same size for all $\mb x$, and it either has a natural ordering or is ordering-invariant. However, 
we allow several relaxations on the size and type on $\nbsvec$ in our framework \revm{as discussed later}.

\subsection{\revm{Semantics of Causal Queries}}\label{sec:causal-queries}
\revm{In this section, we define the semantics of causal queries outlined in Section~\ref{sec:qlanguage} in terms of intervention; how these causal queries are answered in \sys\ using unification of treated and response units, embeddings, and selection of covariates is discussed in Section~\ref{sec:causal-queries:answering}.}
\subsubsection{\revm{Average treatment effect queries}}\label{sec:ate}
The primary causal query in \sys\ is
 average treatment effect (\ate) query \revm{of the form $\outc[\mb X'] \sem \treat[\mb X]?$ as given in (\ref{eq:simple_ATE})}. \revm{Given treatment and response attribute functions $\treat, \outc$, ATE is defined as follows}: 
 \cutr{
 \red{SR: DON'T WE NEED RELATIONAL PATHS AND EMBEDDING HERE? THE FOLLOWING IS NOT SAYING MUCH.}
 \babak{No we don't. The definition is in terms of interventions. Unification and embedding comes later for query answering. }
 }
{      
\begin{align}
\ate(\treat,\outc) \defeq \sum_{\mb x' \in \runit}\frac{1}{m} (\EX[\outc[\mb x'] \mid \Do(\treat[\tunit]=\Vec{0})]- 
\EX[\outc[\mb x'] \mid \Do(\treat[\tunit]=\Vec{1})]) \label{eq:rel-ate}
\end{align} 
}
Intuitively, \ate\ compares the expected response of the response units in two regimes of intervention: one in which all units receive treatment and another where none do. For example, $\ate(\tsc{Prestige},\tsc{Score})$ compares \revm{scores of submissions}
under two interventions in which all authors are and are not affiliated with prestigious institutions.




\subsubsection{\revm{Aggregated response queries}}\label{sec:agg-response-queries}
\revm{Aggregate response queries of the form $\text{\revm{AGG}}\_\outc[\mb X'] \sem \treat[\mb X]?$ as given in (\ref{eq:agg_ATE}) is defined similar to ATE above, where we replace $Y$ with $\text{AGG}\_Y$ everywhere. Note that in the extended relational causal graphs, there are nodes corresponding to $\text{AGG}\_Y$ as shown in Figure~\ref{fig:agg-GCD}. }
\subsubsection{\revm{Relational and isolated effects queries}}
\label{sec:relational-isolated}
The \sys\ query (\ref{eq:relationalquery}) computes the following \revm{three} quantities, which compare the average isolated (\aide), relational (\arlf), and overall (\aoe) effects of two alternative intervention strategies $(\tre,\nbsvec)$ and $(\tre',\nbsvec')$ over $n$ response units\footnote{\revm{We do not need the treatment vectors $\nbsvec, \nbsvec'$ applied to the peers to have the same size although they are applied to all units $\mb x$.} \revm{ We also do not need all units $\mb x$ to have the same number of peers in $\nbs(\mb x)$. As the grammar defined in (\ref{eq:query_grammar}) describes, we can assign treatments to ``at least/most $k$ or $k\%$'' neighbors, and that is well-defined for all units $\mb x$ even if they do not have the exact same number of peers in $\nbs(\mb x)$. On the other hand, for such conditions, we do need to assume that the effects of interventions to the peers are \emph{ordering-invariant}, e.g., the intervention can be applied to any of the $k$ peers (with possible truncations for smaller peer sets) in $\nbs(\mb x)$. }}.
%
%
{     
 \begin{align}
     \aide(\tre; \tre'  \mid \nbsvec) \defeq \frac{1}{n} \sum_{\mb x \in \uunit} \outc_{\mb x}(\tre, \nbsvec)-\outc_{\mb x}(\tre', \nbsvec) \label{eq:aie}
 \end{align}}
%
{      
 \begin{align} \arlf(\nbsvec;\nbsvec' \mid \tre)\defeq  \frac{1}{n} \sum_{\mb x \in \uunit}\outc_{\mb x}(\tre, \nbsvec)-\outc_{\mb x}(\tre, \nbsvec') \label{eq:are}
 \end{align}
}
%
%
{      
 \begin{align}
     \aoe(\tre, \nbsvec ; \tre', \nbsvec') \defeq \frac{1}{n} \sum_{\mb x \in \uunit} \outc_{\mb x}(\tre, \nbsvec)-\outc_{\mb x}(\tre', \nbsvec') \label{eq:aoe}
 \end{align}
}
%
%
Intuitively, the isolated  causal effect of a treatment fixes the treatment of the relational peers of a unit and compares its expected response under two treatment strategies assigned to the unit. On the other hand, the relational causal effect of a treatment fixes the treatment of a unit $\mb x$ and compares its expected response under two treatment strategies assigned to its relational peers. For example, the relational effect of $ \text{Prestige[A]}$ on $\text{AVG\_Score}[A]$ fixes the prestige of an author such as $``Bob"$ and compares the counterfactual response $\text{AVG\_Score}[``Bob"]$ under two regimes of interventions in which the relational peers of $``Bob"$, e.g., $\mathit{``Eva"}$, receive two different treatment strategies, e.g., all of them have prestigious affiliations versus none of them having such affiliations. Note that the overall causal effect is an extension of ATE (\ref{eq:rel-ate}) for two arbitrary treatment strategies. Indeed, ATE coincides with $\aoe(1, \Vec{1} \mid 0, \Vec{0})$ when the treated and response units are unified. 
The following proposition shows the connection between relational, isolated and overall effects \revm{(we omit the proof due to lack of space)}.

\begin{proposition} 
\label{eq:decomposition}
The average overall effect can be decomposed into the average isolated and average relational effects, as follows:
{       
\begin{align}
\aoe(\tre, \nbsvec ; \tre', \nbsvec') & =\aide(\tre,\tre' \mid \nbsvec) + \arlf(\nbsvec,\nbsvec' \mid \tre') \nonumber \\
& =\aide(\tre,\tre' \mid \nbsvec') + \arlf(\nbsvec,\nbsvec' \mid \tre)
\label{eq:rlf}
\end{align} 
}

\vspace{-0.8cm}
\end{proposition}

%% file: probability-model.tex

As discussed in Section~\ref{sec:background}, a causal DAG associates a conditional probability distribution $\pr(X | \Pa(X))$ to each node $X \in \mb V$; these conditional probability distributions are unknown and must be estimated from available data \revm{even for standard causal graphs described in Section~\ref{sec:background}, while there are additional challenges for relational causal graphs.  As described in Section~\ref{sec:language}, in \sys, the relational} causal graph $\gcg$ is obtained by grounding the rules w.r.t. the skeleton database, $\Delta$, and   the number of nodes is not fixed ahead of time but depends on $\Delta$.  
\cutr{
For example, we do not have a single node $\mathit{Score}$, as in Figure~\ref{fig:simple-causal-dag}, but instead have many nodes $\mathit{Score}["s1"]$, $\mathit{Score}["s2"]$, etc. one for each submission in $\Delta$. 
}

\cutr{
We make the reasonable assumption that all nodes that are instances of the same attribute have the same structural equations and thus, the same conditional probabilities.  The assumption is also critical in \sys\ because it makes it possible to estimate the conditional probability distributions from data, and, hence, conduct causal inference.  More precisely, we denote $\fAtt^{\skl} \subseteq \FAtt^{\skl}$ the set of all groundings of an attribute $\fAtt \in \FAtt$ in $\FAtt^{\skl}$, and we 
}

%

\revm{We introduce the following {\em structural homogeneity assumption}, which is critical in \sys\ to estimate the conditional probability distributions from a given observed dataset, and thereby conduct causal inference. Recall that  $\fAtt^{\skl} \subseteq \FAtt^{\skl}$ denotes} the set of all groundings of an attribute $\fAtt \in \FAtt$ in $\FAtt^{\skl}$: 

\begin{itemize}
    \item \textbf{Structural Homogeneity:}  All grounded attributes $\fAtt[\mb x] \in \fAtt^{\skl}$ of the same  attribute $\fAtt \in \FAtt$ share the same structural equation and, hence, the same conditional probability distribution \revm{in equation} (\ref{eq:condl_prob}).
\end{itemize}

For instance, in Example~\ref{ex:univer_schema2}, we assume that all groundings of type $\text{Prestige}$ have the same  structural equations.  

Note that the structural homogeneity assumption concerns only the underlying causal model in relational domains that consist of heterogeneous units; It is fundamentally different from the assumption of homogeneous units made in traditional causal inference (cf. Section~\ref{sec:intro}).

\par
The structural homogeneity assumption, however, is not easily captured
because different groundings of the same attribute can have different number of parents. For instance, consider the atoms $\text{Score}[``s1"]$ and $\text{Score}[``s2"]$ from equation (\ref{eq:decision_s2}).  $\text{Score}[``s1"]$ has two  $\text{Prestige}$ parents (since it has  two authors, \emph{``Eva''} and \emph{``Bob''}), whereas  $\text{Score}[``s2"]$ has one  $\text{Prestige}$ parent (\emph{``Eva."}).  
We address this issue
by using another layer of aggregate functions, \revm{that we call \emph{embeddings}}, $\psi$, and change Equation (\ref{eq:condl_prob}) to
%
\begin{align}
\pr\Big(\fAtt[\mb x] \mid \Emb^{\fAtt}\big(\Pa(\fAtt[\mb x])\big) \Big) \label{eq:cpd}
\end{align}
\noindent where $\Emb^{\fAtt}$ is a collection of mappings that projects the parents of $\fAtt[\mb x]$ into a low-dimensional vector with fixed dimensionality for all $\fAtt[\mb x] \in \fAtt^{\skl}$. Intuitively, we assume that the mappings provide sufficient statistics for evaluating the underling structural questions corresponding to all $\fAtt[\mb x] \in \fAtt^{\skl}$.  More formally, we assume that $\Emb^{\fAtt}$ is known and consists of a set of mappings $\{\emb^{\fAtt}_{\fAtt_1}, \emb^{\fAtt}_{\fAtt_2}, \ldots\}$, one for each type of attribute $\fAtt_j$ occurring on the RHS of a rule (\ref{eq:rse}), where each $\emb^{\fAtt}_{\fAtt_j}$ is an {\em embedding function} that maps the set of values of all parents of type $\fAtt_j$ into a low-dimensional {\em embedding space} with fixed dimensionality.  The embedding function can be \revm{a simple } aggregate \revm{like average}; 
\revm{other types of embeddings are discussed in Section~\ref{sec:embeddings}.}

\cutr{
\babak{Reviewer4 asked for moving the embedding into this section}
}
\begin{example} \em \label{eg:mapping}
  Consider the three nodes of type $\text{Score}$ \revm{for ``$\mathit{s_1}$'', ``$\mathit{s_2}$'', ``$\mathit{s_3}$''} in Figure~\ref{fig:GCD}, and consider their parents of type $\text{Prestige}$.  The number of their parents is 2 \revm{(for ``$\mathit{s_1}$'' -- \textit{``Bob''} and \textit{``Eva''} with vector $\langle 1, 1 \rangle$ for prestige), 1 (for ``$\mathit{s_2}$'' --  \textit{``Eva''} with vector $\langle 1 \rangle$), and 2 (for ``$\mathit{s_3}$'' -- \textit{``Eva''} and \textit{``Carlos''} with vector $\langle 1, 0 \rangle$) respectively (the prestige values of the authors are in Figure~\ref{fig:example_instance}), but under the homogeneity assumption, the conditional probability \cutr{$\pr(\text{Score}| \cdots)$} of scores given prestige of authors would be computed by the same function
  \cutr{.  For that purpose, we first aggregate}
 by using a mapping $\emb^{\text{Score}}_{\text{Prestige}}$ to aggregate the vectors of $\text{Prestige}$ values;} we discuss choices for this aggregate function \revm{in Section~\ref{sec:embeddings}.}
  
\cutr{This results in a single value, which can \sout{be input to the function $\pr(\text{Score}| \cdots)$, the same for all nodes of type $\text{Score}$}
  \revm{compute the conditional probability of score given prestige using the same function for all submissions}.  
 In our running example, the vectors of Prestige values are $\langle 1, 1 \rangle$ for submission $s_1$, $\langle 0 \rangle$ for $s_2$, and $\langle 1, 0 \rangle$ for $s_3$.  The function $\emb^{\mathit{Score}}_{\mathit{Prestige}}$ maps each vector to some value, which is further input to the conditional probability $\pr(\mathit{Score} | \cdot)$.}
\end{example}

To summarize, the grounded causal graph $\gcg$ defined by a relational causal model defines a joint probability distribution given by:
%
%
%
%
\begin{align} \pr(\FAtt^{\skl})&= \prod_{\fAtt[\mb x] \in \FAtt^{\skl} } \pr\Big(\fAtt[\mb x] \mid \Emb^{\fAtt}\big(\Pa(\fAtt[\mb x])\big) \Big) \label{eq:rel-fac} \end{align}

    \revm{In some scenarios, the structural homogeneity assumption may  not hold,  for instance, the structural equations for single-blind and double-blind conferences can be different. Such situations can be expressed in \sys\ by adding multiple rules at {\em different granularities} in which the structural homogeneity assumption is perceived to hold, e.g., 
{\footnotesize
\begin{align}
\text{SBlind\_Score}[S] &\sem \text{Quality}[S] \text{ WHERE } \text{Submission}(S) \nonumber \\
\text{DBlind\_Score}[S] &\sem \text{Quality}[S] \text{ WHERE } \text{Submission}(S) \nonumber
\end{align}
}
}

%% file: queryanswering.tex
\section{Answering Causal Queries}
\label{sec:causal-queries:answering}
\revm{Given the syntax of different causal queries in Section~\ref{sec:qlanguage} and their semantics in Section~\ref{sec:causal-queries}, now we describe how we answer these queries in \sys.}
The query answering component of \sys\ consists of {\em covariate detection} (Section~\ref{sec:cov-detec}) and {\em covariate adjustment} (Section~\ref{sec:cov-adjust}).
The goal of covariate detection is to identify a sufficient set of covariates that should be adjusted for to remove confounding effects.  Then, in the process of covariate adjustment, the data is transformed into a flat, single-table format so that 
\revm{causal inference} can be performed using standard methods. 


\cutr{
 {         
\begin{figure}
    \centering
    \includegraphics[scale=0.55]{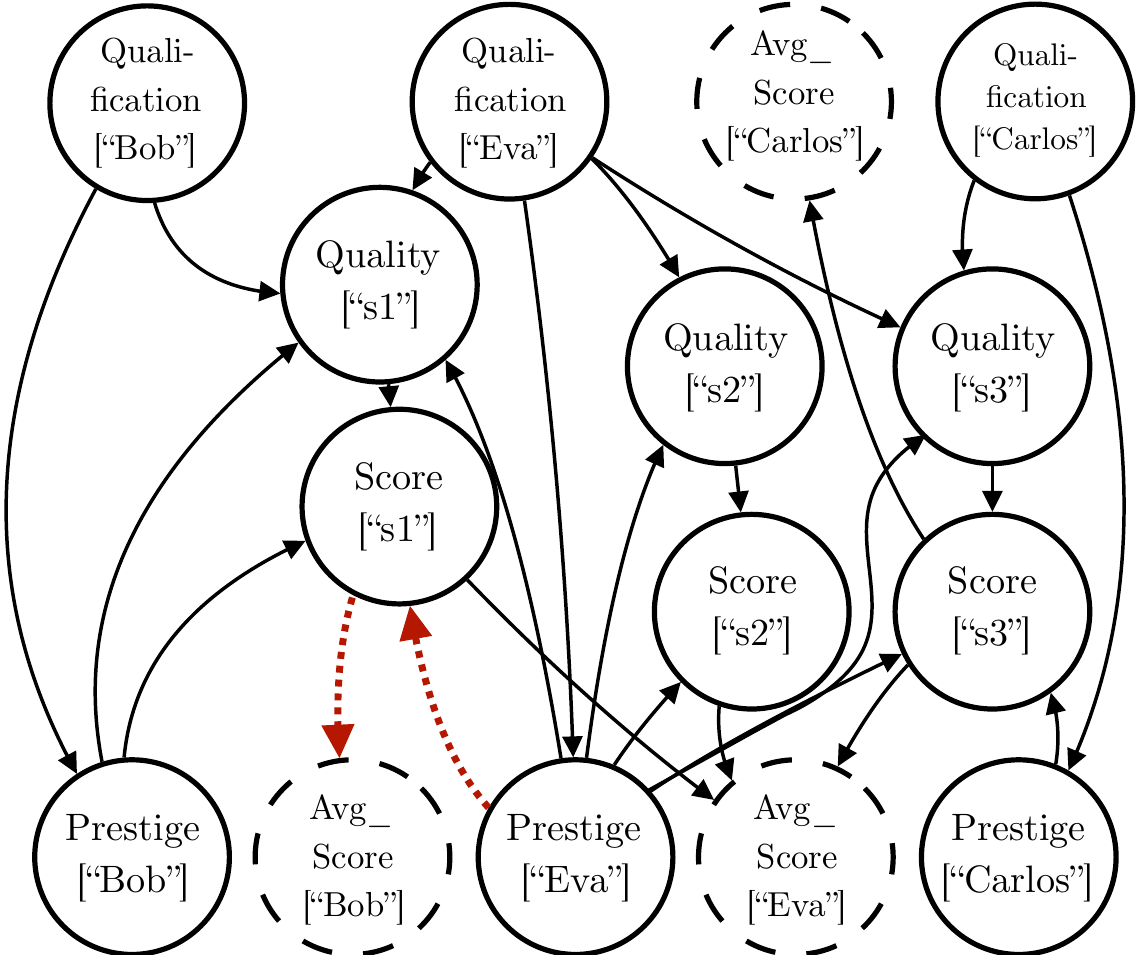}
    \caption{Extending the relational causal graph in Fig~\ref{fig:GCD} with the aggregated attribute $\mathit{AVG\_Score}[A]$.}
    \label{fig:agg-GCD}
\end{figure}
}
}
\subsection{Covariate detection}
\label{sec:cov-detec}

\revm{Given treatment and response attribute functions $\treat[\mb X]$ and $\outc[\mb X']$, to answer all types of causal queries defined in Section \ref{sec:causal-queries}, we need to estimate the effect of interventions
of the form $\Do\big(\treat[\arbtreat]=\trevec_{\arbtreat}\big)$ on a set of treated units $\arbtreat \subseteq \tunit$, on a response unit $\mb x' \in \runit$. This section proves a graphical criterion to select
a sufficient set of covariates from a relational causal graph $\gcg$ that enable the estimation of quantities of the form $\EX\big[\outc[\mb x'] \mid \Do\big(\treat[\arbtreat]=\trevec_{\arbtreat}\big)\big]$ and thereby the queries in Section \ref{sec:causal-queries}. For this purpose we use the extended relational causal graph as shown in Figure~\ref{fig:agg-GCD} to map possibly varying number of parent nodes to a fixed and smaller dimension by adopting the idea of 
embedding functions introduced in Section~\ref{sec:prob-rel-dags}. We illustrate this with an example below.}
\begin{figure}[t]
    \centering
\includegraphics[scale=0.55]{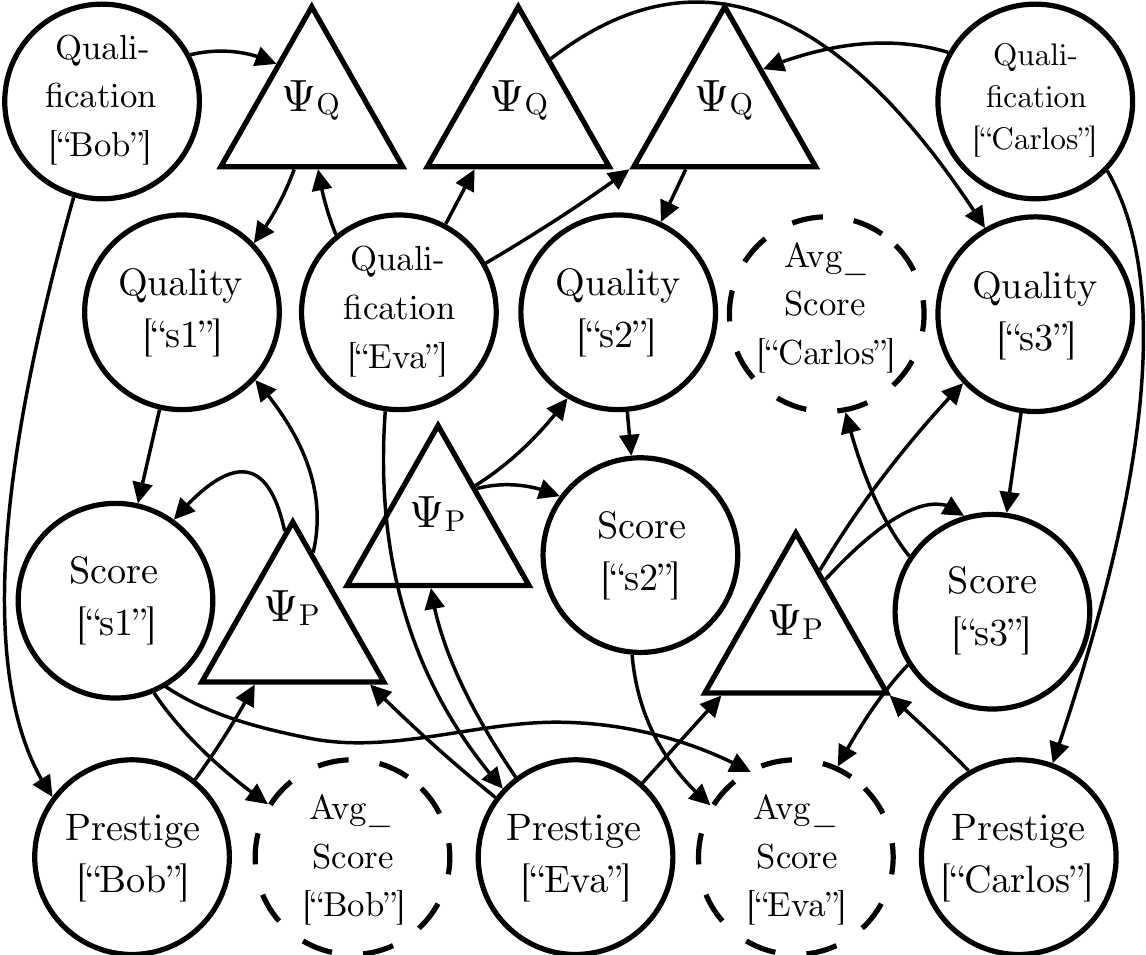}
    \caption{\revm{Final relational causal graph obtained by  (further) augmenting the graph in Figure~\ref{fig:agg-GCD} with embedding functions.} For clarity,  $\Hid^{\mathit{Quality}}_{\mathit{Qualifications}}[S]$ is represented as $\Hid_{Q}$, 
and  $\Hid^{\mathit{Quality}}_{\mathit{Prestige}}[S]$, $\Hid^{\mathit{Score}}_{\mathit{Prestige}}[S]$ 
as $\Hid_{P}$. 
}.
    \label{fig:aug-GCD}
\end{figure}

\begin{example}\label{eg:map-to-attribute}
\revm{In Example~\ref{eg:mapping},
{$\emb^{\mathit{Score}}_{\text{Prestige}}(S)$} 
now corresponds to a new attribute of a submission that maps the $\text{Prestige}$ attribute of {\em Author}s of that submission.
Figure~\ref{fig:aug-GCD} shows the relational causal graph with augmented attributes computed using the mapping functions or embeddings represented by the triangles.}
\end{example}

\par





\revm{The following theorem formalizes how the do-operator for relational causal graph can be estimated from observed data. This theorem uses the concept of \emph{d-separation} from conditional independence in graphical models \cite{PearlBook2000}, denoted by $X \indep Y |_G Z$. The review of these concepts and the proof of the theorem is deferred to the full version of the paper due to lack of space.}
\revm{
\begin{theorem}[Relational Adjustment Formula]  
\label{theore:reladj}
Given an augmented \revm{relational causal graph} $\gcg$, treatment and response attribute functions $T, Y$, and a set of treatment units (entities or relationships) $\arbtreat$ and their treatment assignment vector $\trevec_{\arbtreat}$, we have the following {\em relational adjustment formula}:
\cutr{the quantity  {        $\Ex\big[\outc[\mb x'] \mid \Do\big(\treat[\arbtreat]=\trevec_{\arbtreat}\big)\big]$} the following {\em relational adjustment formula holds}:
}
{        
\begin{align} 
\EX\big[\outc[\mb x'] \mid \Do\big(\treat[\arbtreat]=\trevec_{\arbtreat}\big)\big]& = 
 &  \hspace{-3cm}  \sum_{\mb z \in \Dom(\mb Z)} \
 \EX\big[\outc[\mb x'] \mid \mb Z= \mb z,\treat[\arbtreat']=\trevec_{\arbtreat'} \big] \ \
 \pr(\mb Z=\mb z)
   \label{eq:relational:adjustment}
\end{align}
}
\noindent where $\arbtreat' \subseteq \arbtreat$
is such that, for each unit $\mb x \in \arbtreat'$, there exists a directed path from the node $\treat[\mb x]$ to the node $\outc[\mb x']$ in $\gcg$, and $\mb Z$ is set of nodes in $\gcg$ corresponding to the groundings of a subset of observed attribute functions $\FAttrobs$ such that:
{       
\begin{align} 
\outc[\mb x'] \ \ \ \indep \ \ \ \Big(\bigcup_{\mb x\in \arbtreat}  \Pa\big(\treat[\mb x]) \Big) \ \ \bigr\vert_{\gcg} \ \ \ \Big(\bigcup_{\mb x\in \arbtreat} \treat[\mb x], \ \mb Z\Big) 
\label{eq:raf_idep}
\end{align}    
}
 Further, choosing $Z$ to be the parent nodes of $\arbtreat'$ in  $\gcg$ always satisfies (\ref{eq:raf_idep}) as a sufficient condition.
\label{eq:minimal:indep}
\end{theorem}
}
\vspace{-0.25cm}

\cutr{
\red{SR: directed path is not clear -- what happens if it is a tuple?}

\babak{Both treatment and outcome are atomic nodes.}
}



\ignore{
\revm{\sout{The following observation immediately follows:}}

\begin{observation}\label{obs:parents-condn}
\sout{The $d$-separation (independence) condition (\ref{eq:raf_idep}) is always satisfied for $\mb Z=\bigcup_{\mb x\in \arbtreat}  \Pa\big(\treat[\mb x])$ (i.e., any subset of vertices $d$-separates itself from other vertices).
Hence, adjusting for the joint distribution of the parents of the treated atoms is {\em always} sufficient \revm{for estimation of $\Ex\big[\outc[\mb x'] \mid \Do\big(\treat[\arbtreat]=\trevec_{\arbtreat}\big)\big]$. However, one can use the condition (\ref{eq:raf_idep}) to select a {\em minimal set of observed covariates} for adjustment.}}
\red{SR: This cannot be parsed without a lot of background and details on d-separation}
\end{observation}
}
\revm{(Intuitively, it is always sufficient to condition for the `parents' of treated units as they separate them from the rest of the graph ensuring independence.) }
Here we illustrate with an example.

\begin{example} 
\label{ex:prestigue:identification}
{\em
To compute {        $\ate(\mathit{Prestige},\mathit{Score})$} in our 
example, we need to compute expectations of the form 
{        
\begin{eqnarray}
\revm{\Ex\big[\text{Score}[s] \mid \Do\big(\text{Prestige}[\{``Bob", ``Eva", ``Carlos"\}]=\Vec{t}\big)\big] \  \text{for} \ \Vec{t} \in \{\Vec{0},\Vec{1}\}}  \label{eq:inter}
\end{eqnarray}
}
\noindent where we intervene on all three authors in the example. 
\revm{ By applying Theorem~\ref{theore:reladj} 
 for submission $s=\mathit{``s_1"}$, note that directed paths to $\text{Score}[``s_1"]$ exists only from $\mathit{Prestige[``Bob"]}$ and $\text{Prestige}[``Eva"]$, which form the subset $\arbtreat'$. Further, it is sufficient to condition on the parents of these two Prestige nodes, i.e., $\mb Z=${\small{$ \{\text{Qualifications}[``Bob"], \text{Qualifications}[``Eva"] \}$}}. Therefore,
 (\ref{eq:inter}) reduces to:}
%
{         
\begin{align}
\sum_{\mb z \in \Dom(\mb Z)} 
 \revm{\Ex\big[\text{Score}[``s_1"] \mid \mb Z= \mb z,  \text{Prestige}[\{``Bob", ``Eva"\}]=\Vec{t} \big)\big] \ \pr(\mb Z=\mb z)} \label{eq:s1}
\end{align}
}
\noindent Similarly, for $s=\mathit{``s_2"}$ and $\mb Z=$ {          $ \{\text{Qualifications}[``Eva"]\}$}, we obtain 
{         
\begin{align}
 \sum_{\mb z \in \Dom(\mb Z)} \
  \revm{\Ex\big[\text{Score}[``s_2"] \mid \mb Z= \mb z, \text{Prestige}[\{``Eva"\}]=t \big)\big] \
 \pr(\mb Z=\mb z) } \label{eq:s2}
\end{align}
}

\revm{Note that the relational adjustment formula in (\ref{eq:relational:adjustment}) {\em controls} for an adequate set of
covariates  $Z$ that \emph{confound} the causal effect of a treatment on an outcome ($Z$ is called the set of confounding covariates or covariates). 
For example, the causal effect of  $\text{Prestige}$ on $\text{Score}$
is confounded by $\text{Qualifications}$. This is because, qualified researchers are likely to belong to prestigious universities
and qualified researchers are more likely to submit high quality papers. Therefore, to compute the $\ate$ of  $\text{Prestige}$ on $\text{Score}$ we need to control for author's qualifications as in (\ref{eq:s1})}.
For estimating {        $\ate(\text{Quality},\text{Score})$} (assuming quality is observed) by applying (\ref{eq:raf_idep}) we find that for each submission $s$, {        $\EX\big[\text{Score}[s] \mid \Do\big(\text{Prestige}[\tunit]=\Vec{1}\big)\big]$} can be estimated by adjusting for the embedded attribute functions $\mb Z$ {          $= \{ \Hid^{\text{Score}}_{\text{Prestige}}[s], \Hid^{\text{Score}}_{\text{Qualifications}}[s] \}$}.
}
\end{example}

\begin{sloppypar}
To estimate {        $\ate(\text{Quality},\text{AVG\_Score})$} (the effect on average acceptance rate of an author), on the other hand, we need to estimate {        $\pr\Big(\text{AVG\_Score}[A]=y \mid \Do\big(\text{Prestige}[\tunit]=\Vec{1}\big)\Big)$}, for each author. According to \revd{Equation} (\ref{eq:raf_idep}), this can be done by adjusting for the joint distribution of the qualifications of \revm{{\em all} their \cutr{past}
coauthors, which is  potentially very high-dimensional, and therefore we need another round of embeddings to aggregate that information as discussed next.}
\end{sloppypar}





\begin{algorithm}[t]        
	\DontPrintSemicolon
	\KwIn{ An augmented \revm{relational causal graph} $\gcg$, treated and outcome attribute functions $\treat[\mb X]$ and $\outc[\mb X']$.
}   
	\KwOut{ The unit table $\utable(Y,\mb \emb_{T},\mb \emb_{\mb Z})$.}
		\For{$\mb x' \in \runit$}{
		$\tunit' \gets$ \text{A minimal subset of $\tunit$ such that there exits a } \text{ directed path in $\gcg$ from $\treat[\mb x]$ to $\outc[\mb x']$ for all $\mb x \in  \tunit'$} \\
		$\mb Z \gets \text{A minimal set of vertices in $\gcg$}$ \text{that satisfies the $d$-separation statement in Eq~(\ref{eq:raf_idep})} \\
	 $\emb_{\treat} \gets \emb^{\outc}_{\treat}(\langle \treat[\mb x_1], \ldots,  \treat[\mb x_{|\tunit'|}] \rangle)$\\ 
	 	 $\Emb_{\mb Z} \gets \Emb^{\mb \outc}_{Z}(\mb Z)$  \\
	Insert the tuple $(Y[\mb x],\mb \emb_{T}[\mb x],\mb \emb_{\mb Z}[\mb x]))$ to unit table $\mb D$ 
		}    
	\caption{        Constructing a unit table.} \label{algo:unit-table}
\end{algorithm}

\subsection{Covariate adjustment}
\label{sec:cov-adjust}
\revm{There are two challenges in estimating the causal queries in Section~\ref{sec:causal-queries} using the relational adjustment formula (\ref{eq:relational:adjustment}): (1) when the set of confounding covariates $\mb Z$ has high dimensionality, estimating the conditional expectation in (\ref{eq:relational:adjustment}) from data is challenging, and (2) the causal queries need to compute {\em averages} across all response units. Hence, we need to estimate the formula (\ref{eq:relational:adjustment}) separately for each response unit that is not feasible. For instance, in Example~\ref{ex:prestigue:identification}, (\ref{eq:s1}) and (\ref{eq:s2}) need to be estimated separately.
}

\par

\revm{To address these issues, similar to Section~\ref{sec:prob-rel-dags}, we use a set of embedding functions $\emb^{\outc}_{\treat}$ and $\Emb^{\outc}_{\mb Z}$ to project the treatment and covariate vectors, respectively, into a low-dimensional embedding space with {\em fixed dimensionality} and for all response units. \cut{Note that the homogeneity assumption in Section \ref{sec:causal-framework} implies that the embedding representations of the treatment and covariate vectors \reva{have the same distribution}.} This enables  us to transform a (multi-) relational instance to a single low-dimensional flat table.}

\subsubsection{Unit table} \revm{In the classical causal inference framework model discussed in Section~\ref{sec:background}, the units of interest are stored in a single unit table with attributes corresponding to the treatment, response, and confounding covariates as the columns. Here we generalize this concept to capture units in relational causal analysis.}
\par
Given a \revm{relational causal graph} $\gcg$ and  treatment and response attribute functions $\treat[\mb X]$ and $\outc[\mb X']$, we use Algorithm~\ref{algo:unit-table} to construct a {\em unit table}, which is a standard relation (table) with schema $\utable(\outc,\mb \emb^{\outc}_{\treat}, \Emb^{\outc}_{\mb Z})$ \revm{(note that $\Emb^{\outc}_{\mb Z}$ denotes a vector of values for possible multiple covariates $\mb Z$)}. It consists of tuples $(\outc[\mb x'],\emb^{\outc}_{\treat}[\mb x'], \Emb^{\outc}_{\mb Z}[\mb x'])$ for each response unit $\mb x' \in \runit$, where $\emb^{\outc}_{\treat}[\mb X']$ and  $\Emb^{\outc}_{\mb Z}[\mb X']$ (with abuse of notation) are relational embedded attribute functions that correspond to the result of applying $\emb^{\outc}_{\treat}$ and $ \Emb^{\outc}_{\mb Z}$ to the treatment and covariate vectors respectively.
%

\begin{table}
\centering          
\begin{tabular}{|c|c|p{3.0cm}|p{3.0cm}|p{3.0cm}|}
    \hline
             \multicolumn{1}{|c|}{\bfseries Unit} &  
             \multicolumn{1}{|c|}{\bfseries Outcome ($\outc$)} &  
     \multicolumn{1}{|p{2.5cm}|}{\bfseries Embedded coauthors' treatments ($\emb^{\outc}_{\treat}$)} &  
    \multicolumn{2}{|p{4cm}|}{\bfseries Embedded Collaborators' Covariates ($\Emb^{\outc}_{\mb Z}$)}\\  \hline
      Author ID&     {\normalfont \text{AVG}}\_Score  & Prestige (AVG) & Centrality (COUNT) &  H-index (AVG)   \\ \hline 
            Bob&    0.75  & 1 & 1 & 2 \\  \hline
                        Carlos &  0.1  & 1 & 1 & 2 \\  \hline
                                    Eva & 0.41 & 0.5 & 2 & 35 \\  \hline
\end{tabular}
\caption{The unit table for $\treat[\mb X] = \mathit{Prestige[A]}$ and  $\outc[\mb X'] = \text{AVG}\_Score[A]$ based on Figure~\ref{fig:example_instance}. } 
 \label{fig:unittable:author}
\end{table}

\begin{example} {\em Table~\ref{fig:unittable:author} shows the unit table corresponding to  $\treat[\mb X] = \text{Prestige[A]}$  and  $\outc[\mb X']$ =  $\text{AVG\_Score}[A]$. \revm{Here {\em Authors} constitute the response units and the aggregated response is an attribute of authors}.  In this table simple mappings such as average and count are used for embedding.  \revm{Note that Table~\ref{fig:unittable:author} also serves as the unit table for $\treat[\mb X] = \text{Prestige[A]}$  and  $\outc[\mb X']$ =  $\text{Score}[A]$. In this case since the treated and response units are different \sys\ uses the aggregated response $\text{AVG\_Score}[A]$ for unification (see Section~\ref{sec:rel-paths-peers}).}
}
\end{example}
By rewriting the RHS of the relational adjustment formula (\ref{eq:relational:adjustment}) in terms of the  attributes of the unit table and $\trevec_{\arbtreat'}^e$ the embedded representation of the treatment assignment $\trevec_{\arbtreat'}$, (\ie, $\trevec_{\arbtreat'}^e=\emb^{\outc}_{T}(\trevec_{\arbtreat'})$), we obtain
{        
\begin{align} 
\sum_{\mb z \in \Dom(\mb \Emb^{\outc}_{\mb Z})} \
 \revm{\Ex\big[Y \mid  \Emb^{\outc}_{\mb Z}= \mb z,\emb^{\outc}_{\treat}= \trevec_{\arbtreat'}^e \big] \
 \pr(\mb \Emb^{\outc}_{\mb Z}=\mb z)}
   \label{eq:relational:adjustment:embeded}
\end{align}
}
\revb{Once we have a flat unit table with columns for treatment, response, and covariates as in Section~\ref{sec:background}, the causal queries in Section~\ref{sec:causal-queries} can be estimated using 
(\ref{eq:relational:adjustment:embeded}) by applying the standard approaches to causal analysis like regression \cite{angrist2008mostly}
(the conditional expectation in (\ref{eq:relational:adjustment:embeded}) is a regression function) or matching methods \cite{ho2007matching, Rosenbaum93, iacus2009cem} (matching treatment and control units with the same/similar values).} \reva{The validity of treatment effect estimates is conditional on the assumption that the background knowledge is accurate.}


\subsubsection{\revd{Choice of embedding functions}}\label{sec:embeddings} Embedding as a technique addresses both the issues of the high dimensionality and the variable size of the treatments and covariates that correspond to the response units, thereby making the estimation of causal queries more convenient. However, (\ref{eq:relational:adjustment:embeded}) only approximates (\ref{eq:relational:adjustment}), hence the quality of the answers depends on  whether the embeddings preserve sufficient statistics. 
\revd{In this work, we use the following natural choices of embeddings, a formal study of the choices of embedding functions in multi-relational causal analysis is an interesting direction of future work.
    (1) {\em Mean} and {\em median:} Uses basic aggregation functions, such as mean and median, together with the cardinality of the vectors (to account for the underlying topology of the relational skeleton, e.g, number of authors or collaborators). 
    (2) {\em Padding:}  Pads each variable size vector with out-of-band ``empty markers" to make create same-sized vectors to use directly as the embedding.
    (3) {\em Moments:} Uses a vector consists of $k$ moments (\ie, mean, variance,  skewness, etc.), where $k$ is chosen to minimize response prediction loss.
}
\vspace*{-3mm}

%% file: Experiment.tex
\section{Experiments}

\label{sec:experiments}
\allowdisplaybreaks

\begin{figure}[!tbp]
    \centering
    \subfloat[]{\includegraphics[width=0.48\textwidth]{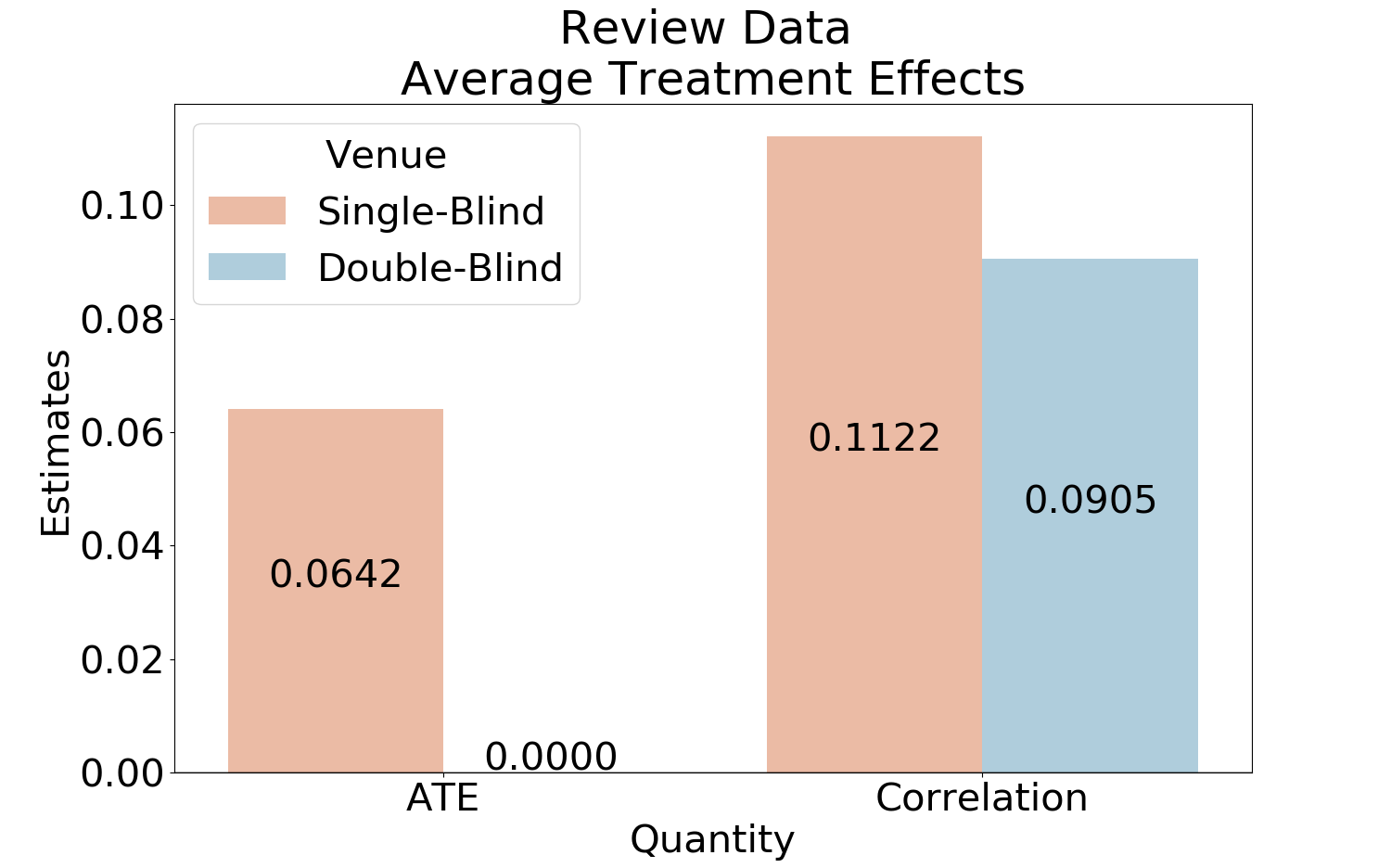}}
    \quad
    \subfloat[]{\includegraphics[width=0.48\textwidth]{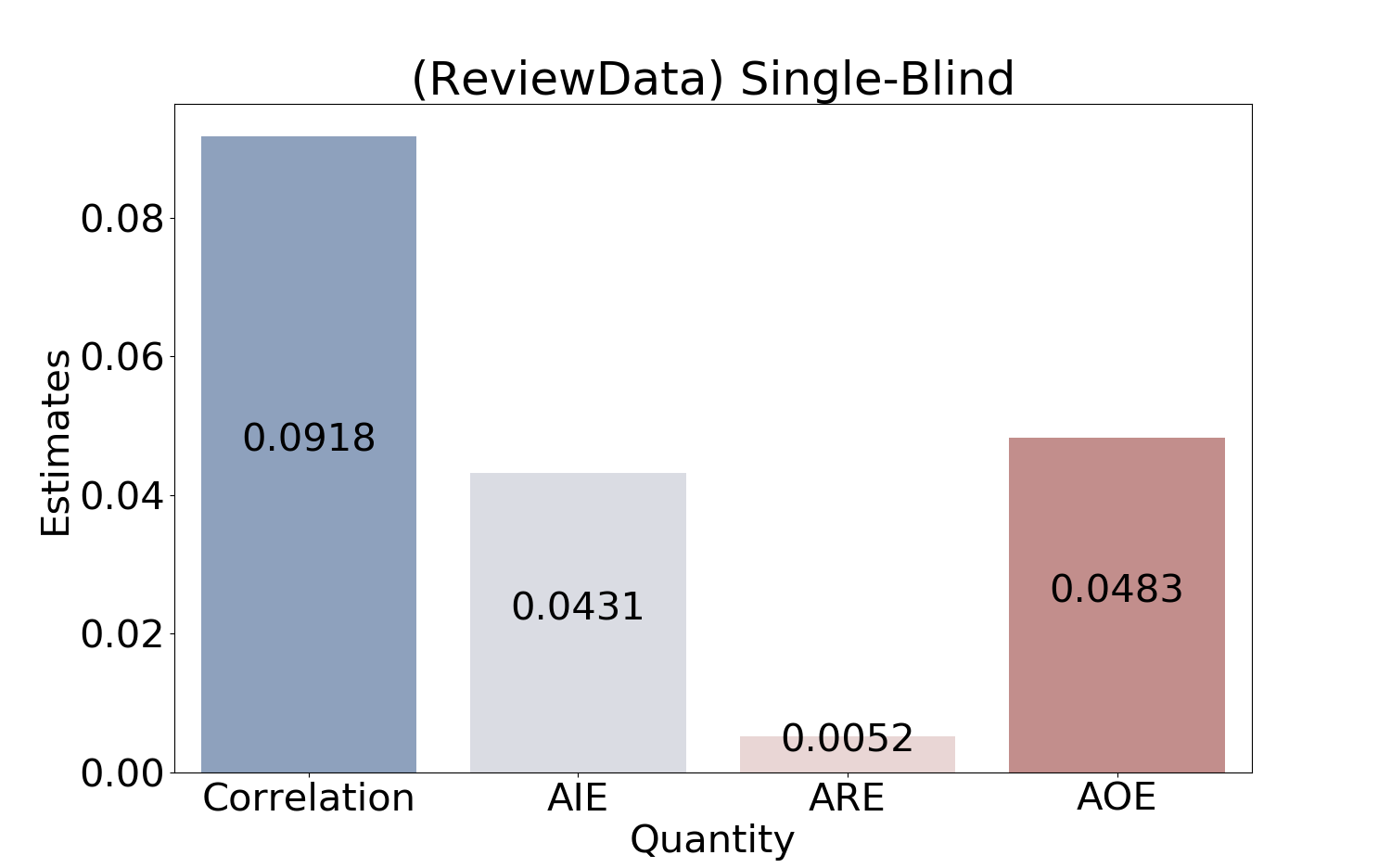}}
    \caption{(a) The average treatment effect estimates and Pearson's correlation for single-blind and double-blind submissions. (b) Pearson's correlation, average isolated effect, average relational effect and average overall effect for all authors on submissions in single-blind venues.}
    \label{fig:end2end}
\end{figure}
\cutr{
\revm{This section discusses the experiments to evaluate the feasibility and efficacy of \sys.  We used real data from three unique datasets from different domains to understand \sys's end-to-end performance. These real datasets are (a) \data\ \cite{OpenReview, scopus, shanghai} (an integrated review dataset as used in Example~\ref{ex:univdomain:causalmodel}, Section~\ref{sec:expt-reviewdata}) (b) \mimicdata\ \cite{mimic} (the Multiparameter Intelligent Monitoring in Intensive Care dataset, Section~\ref{sec:expt-mimicdata}), and (c) \nisdata\ \cite{healthcare_cost_and_utilization_project_hcup_hcup_2006} (the Nationwide Inpatient Sample dataset, Section~\ref{sec:expt-nisdata}), described in detail in their respective sections. Since in any observed real dataset ground truth for the treatment effect in causal analysis is unknown, we
also used a synthetic review dataset following \data\ with known ground truth treatment effects to validate: (1) \sys's end-to-end performance and comparison with standard causal analysis on a combination (join) of multi-relational data}
\cutr{
(particularly 
with regards to avoiding ignoring the relationality of data and naively conducting a causal analysis on the table obtained by joining all base relations)
}, (2) the recovery of the underlying true treatment effect using \sys, and (3) the sensitivity of query answering using \sys\ for different embedding methods.
}

\revm{

  In this section, we conduct an experimental evaluation of \sys,
  addressing three questions. \textbf{End-to-end performance}: is
  \sys\ effective in answering causal queries on relational data? Can
  it avoid simply discovering correlations instead of true causation?
  Can it distinguish isolated effects from relational effects?
  \textbf{Quality of estimates}: when ground truth is available, can
  \sys\ recover the true treatment effects?  And is the relational
  structure necessary for recovering the correct treatment effect?
  \textbf{Sensitivity to embeddings}: how sensitive is \sys's
  performance to the choice of the type of embedding strategy?

\textbf{Experimental Setup.~}
The experiments were performed locally on a 64-bit Linux server with 1TB RAM and 4 Intel Xeon processors with 15 cores @ 2.8GHz each. }
\revm{
\subsection{Datasets} \label{subsec:datasets}
\revm{{
\begin{table}[] \centering
		\revm{
		\begin{tabular}{@{}lcccccc@{}}\toprule
		    {Dataset} & Tables [$\#$] & {Att. [$\#$]} & {Rows [$\#$]} &Unit Table Cons. & Query Ans.\\ \midrule
		    \textbf{\mimicdata} & 26 & 324 & 400M & 6h & 4.5h \\ \hdashline
		    \textbf{\nisdata} & 4 & 280 & 8M & 4m & 30s \\ \hdashline
			\textbf{\data} & 3 & 7 & 6K &  10.6s & 1.2s \\ \hdashline
			\textbf{\sdata} & 3 & 7 & 300K & 17.2s & 1.3s \\ \hdashline
		\end{tabular}}
		\caption{\revm{Data description and query runtime.}}
		\label{tab:data}
\end{table}
}
}
We used three real datasets, two from the medical domain, and one
about conferences, summarized in Table~\ref{tab:data}.  All datasets
contain interesting relationships that inform \sys's causal analysis.
In addition, we generated a synthetic dataset in order to have control
over the ground truth.

\paragraph{\bf \mimicdata} The Multiparameter Intelligent Monitoring in
Intensive Care III (\mimicdata) database is a large-volume,
multi-parameter dataset collected from the ICUs of Beth Israel Deaconess
Medical Center from 2008 to 2014 representing 38,597 adult patients,
58,976 hospital admissions \cite{mimic}.  There are 26 tables with 400M rows and
324 attributes (see Table~\ref{tab:data}),
which include patients' information like demographics, length of stay,
medications, laboratory test results and health
insurance data. 
We specified the following causal model in \sys, where  \Ethnicity\ = ethnicity, \Patient\ = patient, \Diagnosis\ = diagnosis, \Doctor\ = doctor, and \los\ = length of stay:

\revm{
{
\begin{align*} 
\SelfPay[P] &\sem \Ethnicity[P], \Religion[P], \Gender[P] \text{ WHERE } \Patient(P)\\
\Diagnosis[P] &\sem \Ethnicity[P],\Religion[P],\Gender[P] \text{ WHERE } \Patient(P)\\
\DrugDose[D] &\sem \Diagnosis[P], \Severity[P], \Doctor[C]   \text{ WHERE }  \Prescribed(C, D), \Caregiver(C, P) \\
\Death[P] &\sem \los[P],\Diagnosis[P], \DrugDose[D], \Doctor[C] \\
& \text{ WHERE } \Caregiver(C, P), \Administered(D, P)\\
\los[P] &\sem \DrugDose[D],\Diagnosis[P] \text{ WHERE }  \Administered(D, P)
\end{align*}
 }
 }

\paragraph{\bf \nisdata} The Nationwide Inpatient Sample (\nisdata) \cite{healthcare_cost_and_utilization_project_hcup_hcup_2006} is a dataset of hospital stays across the US, produced by the Department of Health and Human Services once annually. We use the sample for the year 2006, which comprises 8 million hospital admissions across 1035 hospitals. Each admission is associated with a hospital and the patient's demographic information, admission source, health history, performed procedures, and new diagnoses. Information available about each hospital includes size, location, and ownership.
We specified a casual model in \sys\ using  16 intuitive causal rules,
using attributes whether the hospital is classified as large \cite{agency_for_healthcare_research_and_quality_bedsize_nodate}, patient's medical bill, etc.; we mention a few below:
%
{
\begin{align*}
\text{Bill}[P] & \sem \text{Illness\_Severity[P]} \\
\text{Bill}[P] & \sem \text{Private\_Ownership[H]} \text{ WHERE } \text{Admitted}(P, H) \\
\text{Bill}[P] & \sem \text{Surgery\_Performed[P]}\\
\text{Admitted\_to\_large}[P] & \sem \text{Illness\_Severity[P]}
\end{align*}
}
}
\paragraph{\bf \data.}\label{sec:review-data}
 \data\ consists of 2,075 papers submitted for review between 2017 and 2019 at 10 computer science conferences and workshops, which have acceptance rates between 40\%--84\%. Each submission is associated with a number of referee reviews and an acceptance or rejection decision. About half of all submissions are double-blind, while the other half reveal author names to the reviewer. All submissions were unblinded after the conferences concluded. The dataset also contains an authors table, with the citation count, h-index, publishing experience (in years), and university ranking for each of the 4490 authors who contributed to a paper in the dataset.  \data\ was built by scraping, cleaning and normalizing data from OpenReview~\cite{OpenReview}, Scopus~\cite{scopus} and the Shanghai University Rankings~\cite{shanghai}. Scraping Scopus was done using the tool proposed in ~\cite{rose_pybliometrics:_2019}. We plan to make \data\ publicly available.

\paragraph{\bf \sdata.}
We generated  \sdata\
 mimicking the probability distributions observed in \data. The relational skeleton was generated keeping in mind the correlations we observed in the real data, \eg, authors with high productivity tend to be affiliated with more prestigious institutions, and authors from more prestigious institutions tend to collaborate with each other more. However, for each paper we let the number of authors and each submission's choice of venue be determined randomly. We generated $10,000$ authors with affiliations to $200$ different institutions, along with $75,000$ papers submitted to $100$ different venues. Next, we generated two datasets to explore \sys's performance with and without relational effects. The first dataset had a treatment effect of prestige on review score of 0 for double-blind and 1 for single-blind venues, for all submissions. In the second dataset, the isolated effects stay the same for both double- and single-blind venues while there is a constant effect of $1/2$ on the review score of each submission if authors' collaborators are prestigious.

\cutr{\textbf{Result Summary.~}\label{sec:resultsummary} The experiments on \data, \mimicdata~ and \nisdata~ displays \sys's end-to-end ability to setup the relational causal schema, encode background knowledge and answer varied causal queries on data from different domains. On \data, we answer two sets of causal queries. The first set investigates the effect of prestige of paper's authors on paper's review scores. The second set of queries investigates the isolated, relational and overall effects of an author's prestige on their average chance of paper acceptance. We compare these results with naive correlation to highlight the contrast between causal effect and association. Analysis on \mimicdata~ highlights that while there is minuscule causal linkage between insurance coverage and the quality of care provided, we still observe stark difference when we just look at associations. The showcases that the \textit{a priori} differences between the groups with and without insurance is reflected their medical outcome. We see a similar reversal of trend for \nisdata, where we look at the effect of the size of hospital on cost of treatment. These examples illustrates the differences between naive association based inference and inferences based on causal analysis. 
}
\begin{figure}
    \centering
    \includegraphics[width=0.5\textwidth]{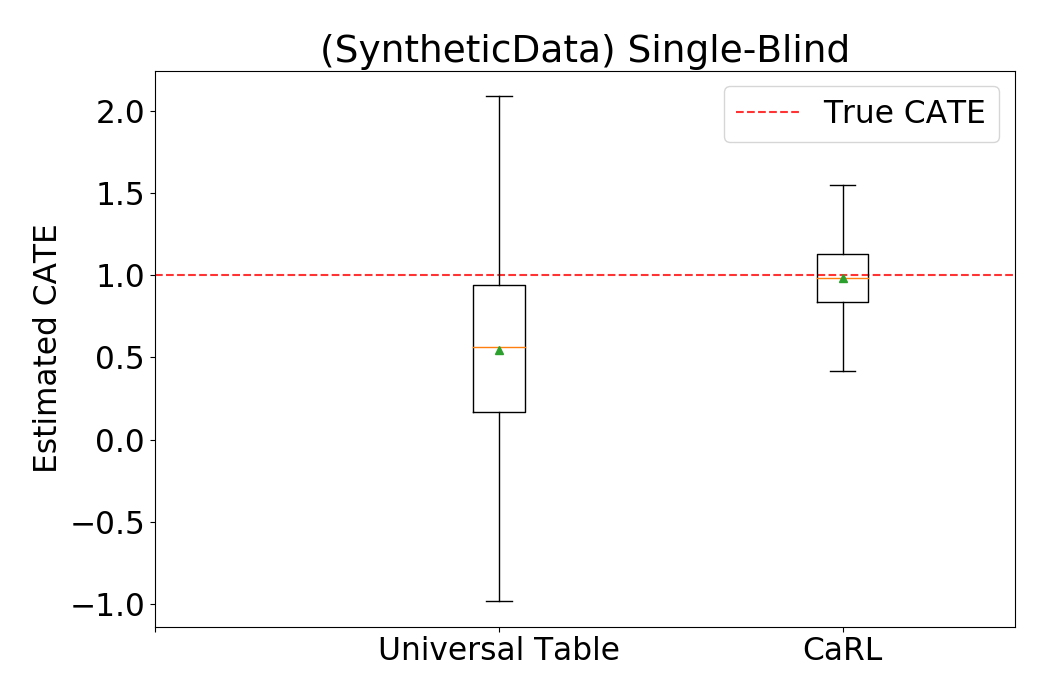}
    \caption{Comparing CATEs estimated using the universal table obtained by joining all base relations and \sys.}
    \label{fig:unitable}
\end{figure}
\subsection{End-to-end results}

\label{sec:exp_end_to_end}


In this section we used \sys\ to answer several causal queries (including all kinds defined in Section~\ref{sec:qlanguage}) on the
real datasets and evaluated their quality.  Since we do not have ground
truth for this data, we discuss which results are more in agreement
with the intuition or the literature in the field.  We also compared
\sys's answer with more naive, correlation-based answers.  \reva{All
  runtimes for these experiments are reported in
  Table~\ref{tab:data}.}

\revm{\paragraph{\bf \mimicdata} We asked the following causal
  queries: what is the effect of not having health insurance on the
  mortality rate? And what is the effect of non-insurance on the
  length-of-stay (in the hospital)?

{          
\vspace*{-3mm}
\begin{align}
(a)~ \Death[P] &\sem \SelfPay[P] ? 
&
(b)~ \los[P] \sem \SelfPay[P] ?
 \label{eq:query_LOS}
\end{align}
\vspace*{-5mm}
}

The treated and control groups consist of patients without insurance
(self-payers) and with insurance
respectively. Table~\ref{tab:mimic-res} shows the results for both the
average treatment effect (ATE) and the naive difference of averages
between the two groups.  Computed naively, there is a significant
difference in both mortality rate and the length of stay between
insured and non-insured groups. However, after adjusting for
confounders and mediators, we observe that there is almost no effect
on mortality rate; in other words, care givers do not discriminate in
treating patients with and without insurance.  The discrepancy is due
to the fact that self-payers tend to defer checking into a hospital
until the problem is severe.  The treatment effect on the length of stay is also
attenuated compared to the estimated difference between the average of the outcomes for treated and control groups.

{          
\begin{table}[H]
\centering
\resizebox{0.5\textwidth}{!}{%
\begin{tabular}{|l|c|c|c|c|}
\hline
\textit{\textbf{\begin{tabular}[c]{@{}l@{}}Causal Query\end{tabular}}} & \multicolumn{1}{l|}{\textbf{\begin{tabular}[c]{@{}l@{}}Avg. of\\ Treated\end{tabular}}} & \multicolumn{1}{l|}{\textbf{\begin{tabular}[c]{@{}l@{}}Avg. of\\ Control\end{tabular}}} & \multicolumn{1}{l|}{\textbf{\begin{tabular}[c]{@{}l@{}}Diff. of\\ Averages\end{tabular}}} & \multicolumn{1}{l|}{\textbf{ATE}} \\ \hline
\textbf{MIMIC~1 (\ref{eq:query_LOS}-a)} & 15.5\% & 9.8\% & 5.7\% & 0.5\% \\ \hline
\textbf{MIMIC~2 (\ref{eq:query_LOS}-b)} & 154.23h & 244.15h & -89.92h & -26.04h \\ \hline
\textbf{\nisdata~1 (\ref{eq:query_sizebill})} & 64\% & 31\% & 33\% & -10\% \\ \hline
\end{tabular}%
}
\caption{\revm{The Average Treatment Effect (ATE) compared to naively
    computing the difference between the averages of the treated and
    control groups.
}}
\label{tab:mimic-res}
\end{table}
}
\paragraph*{\bf \nisdata} We asked the following causal query: are
patients admitted to large hospitals charged more than those admitted
to small hospitals?  Expressed in \sys, the query is:
%
\begin{align}
\text{AVG\_Bill}[H] &\sem \text{Admitted\_to\_Large}[P]? \label{eq:query_sizebill}
\end{align}
%

The treated and control groups are large and small hospitals
respectively.  As before, we compared the ATE with the naive
difference of the average bills of the two groups, and show the
results in Table~\ref{tab:mimic-res}.  The naive computation shows
that the average bill at large hospitals is 33\% more likely to be larger per patient
(\textit{i.e.,} less affordable).
%
However, when computing the ATE, \sys\ adjusts for the profile of the
patients each hospital receives, and we obtain a surprising reversal
of the trend. The reason for this discrepancy is that patients with
more severe (and, thus, more costly) conditions tend to go to large
hospitals, while small hospitals tend to have patients with milder
conditions.  In fact, the medical literature reports that, all
else being equal, a larger hospital will provide more affordable
treatment than a small one.  One meta-analysis
\cite{giancotti_efficiency_2017} reports that economies of scale are
present in the healthcare sector and so finds support for the policy
of several national governments to consolidate smaller hospitals to
increase productivity and efficiency.}
\revm{\paragraph{\bf \data} We asked two casual queries: what is the
  effect of an author's {\em prestige} on the average score of his/her submissions?  And what
  is the effect on the submission score when more than 1/3 of her
  co-authors are treated?  Expressed in \sys, the queries are:
{          
\begin{align}
& \text{AVG\_Score}[A] \sem \text{Prestige}[A]?\label{eq:review:query:b}\\
& \text{Score}[S] \sem \text{Prestige}[A]? \  \mathtt{WHEN} \  \mathtt{MORE  \ THAN} \ 1/3  \  \ \mathtt{PEERS} \ \mathtt{TREATED}\label{eq:review:query:c}
\end{align}
}
}
\cutr{Notice that these cover all three query types supported by \sys\ as
described in Section~\ref{sec:qlanguage}.}  We ran each query twice,
once on single-blind conferences, and once on double-blind; in \sys,
this is achieved by adding a \texttt{where} condition to the queries
(not shown here), and computed the $\ate$ in both cases.
%
%
\revm{In addition, we also computed the Pearson correlation between
  the score distributions of prestigious and non-prestigious authors.
  The results are shown in Figure~\ref{fig:end2end}(a),
} and show a significant correlation, both for single-blind and
double-blind conferences. However, \sys\ found that the causal effect
of prestige on submission scores was {\em significant} for
single-blind venues, but {\em not significant} for double-blind
venues. A naive interpretation of correlation-as-causation leads to
the false conclusion that double blinding is not effective in reducing
bias.  While the validity of our these findings depend on the validity
of the underling assumptions made in this paper, we believe they
surpass naive correlation. In particular, we note that our results are
in accordance with a series of controlled experiments that suggest
double-blind reviewing does indeed reduce institutional prestige
bias~\cite{ross2006effect,snodgrass2006single,okike2016single,tomkins2017reviewer}.

Given its primarily networked structure, \data\ offers a great
opportunity to compute peer effects.  (\reva{In contrast, there are no
  relational peers for the causal queries on \mimicdata\ and
  \nisdata.})  We computed the effect of prestige across peers on
review scores in single-blind conferences, and used \sys\ to compute
the isolated, the relational, and the overall effects as in
\eqref{eq:review:query:c}. Figure~\ref{fig:end2end}(b) reveals that
the isolated effect ($\aide$) is more significant than the relational
effect ($\arlf$), meaning that an author's own prestige has a stronger
effect on his or her average submission score than their
collaborators' prestige has, as we might expect. Furthermore, one can
verify that we obtained $\aoe=\aide+\arlf$, which independently
conforms with \revm{Proposition~\ref{eq:decomposition}}.

\subsection{Quality of estimates}

{          
\begin{table}
\centering
\revm{
\resizebox{0.45\textwidth}{!}{%
\begin{tabular}{|l|l|c|c|c|}
\hline
\multicolumn{2}{|l|}{~}&
\multicolumn{1}{l|}{\textbf{AIE}} & \multicolumn{1}{l|}{\textbf{ARE}} & \multicolumn{1}{l|}{\textbf{AOE}} \\ \hline
\multirow{2}{*}{\textbf{Single-Blind}} & \textbf{Estimated} & 1.138 & 0.434 & 1.573 \\ \cline{2-5}
~ & \textbf{Ground Truth} & 1.000 & 0.500 & 1.500 \\ \hline
\multirow{2}{*}{\textbf{Double-Blind}} & \textbf{Estimated} & 0.101 & 0.429 & 0.538 \\ \cline{2-5}
~ & \textbf{Ground Truth} & 0.000 & 0.500 & 0.500 \\ \hline
\end{tabular}%
}
}
\caption{\revm{Averages for isolated, relational and overall effects for \sdata\ by query in (\ref{eq:review:query:b}).}}
\label{tab:rel_pdf_single}
\end{table}
}

As the ground truth is not known for the real datasets, we use \sdata\ to evaluate the quality of the estimates \sys\ provides.

We report estimated and true $\ate$, $\arlf$, $\aide$ and $\aoe$ to scrutinize \sys's performance. \revm{ As seen in Table~\ref{tab:rel_pdf_single}, \sys\ is able to disentangle the isolated and relational effects present in \sdata. It is able to do so for both sub-populations, which have different generative rules. The different estimates are correctly recovered, and the property $\aoe=\aide+\arlf$ from Proposition~\ref{eq:decomposition} is again respected.

\begin{table}[]
\centering
\resizebox{0.6\textwidth}{!}{%
\revm{
\begin{tabular}{|l|l|c|c|c|c|}
\hline
\multirow{2}{*}{\textbf{Method}} & \multirow{2}{*}{\textbf{Embedding}} & \multicolumn{2}{c|}{\textit{\textbf{Single-Blind}}} & \multicolumn{2}{c|}{\textit{\textbf{Double-Blind}}} \\ \cline{3-6}
 & & \textbf{Estimated} & \textbf{True} & \textbf{Estimated} & \textbf{True} \\ \hline
\multirow{4}{*}{\textbf{\sys}} & \textbf{Mean} & $1.124\pm0.43$ & $1.00$ & $0.192\pm0.40$ & 0.00 \\ \cline{2-6}
 & \textbf{Median} & $1.119\pm0.36$ & 1.00 & $0.115\pm0.37$ & 0.00 \\ \cline{2-6}
 & \textbf{\begin{tabular}[c]{@{}l@{}}Moment\\ Summary\end{tabular}} & $1.020\pm0.36$ & 1.00 & $0.109\pm0.32$ & 0.00 \\ \cline{2-6}
 & \textbf{Padding} & $1.011\pm0.29$ & 1.00 & $0.013\pm0.30$ & 0.00 \\ \hline
\textbf{\begin{tabular}[c]{@{}l@{}}Universal\\ Table\end{tabular}} & \textbf{N/A} & $0.54\pm0.73$ & 1.00 & $0.201\pm0.64$ & 0.00 \\ \hline
\end{tabular}%
}
}
\caption{\revm{Comparing the sensitivity of the quality of query answer to different choice of embeddings on \sdata, using the query in (\ref{eq:review:query:c}).}}
\label{tab:cate_violin}
\end{table}

To test the ability to utilize relational structure, we computed the treatment effect estimates to the causal queries \eqref{eq:review:query:c} using \sys\ and compared to propensity score matching on the universal table obtained by joining all base relations.}
\revm{Table~\ref{tab:cate_violin} compares the estimates by these two approaches with the ground truth. As shown, in all tested cases \sys\ approximately recovered the ground truth within a reasonable error bound.} However,  causal inference on the universal table resulted in an {\em incorrect ATE with a considerable variance.} This experiment reveals that ignoring the relational structure in relational domains can lead to incorrect estimates and erroneous conclusions.

\subsection{Sensitivity to embeddings}
\revm{
Assessing the effect of embeddings requires access to the ground truth, so we restrict ourselves to testing on \sdata\ in this subsection. Table~\ref{tab:cate_violin} shows that while \sys\ consistently recovers the ATE, the correct choice of embedding can improve its performance. We observe that simple embeddings (such as mean or median) recovered approximately the true average treatment effect. However, their estimate was less centered around the ground truth compared to embeddings like padding or moment summarization. While padding had the tightest variance,  moment summarization also showed promising results. These trends apply regardless of whether we consider single- or double-blind venues, each of which has different generative models and ground truths. It is important to note that moment summarization is one of the simpler approaches for set embedding\footnote{More sophisticated approaches exist, \eg, recurrent neural networks or kernel density estimators. We leave these for future work.}. Additionally, the padding technique tends to create vectors that grow in proportion to the size of the relational skeleton, which limits to its applicability.
}

\input{expt-synthetic}

\input{mimic}

\input{nis}

\cut{
\revm{Based on the results of the above mentioned experiments, we conclude that \sys~ can recover true treatment effects on relational data from multiple domains at granular level and helps the investigator to dissociate causal effects from correlations.}
}

%% file: expt-synthetic.tex
\begin{figure}
    \subfloat[]{\includegraphics[width=0.5\textwidth]{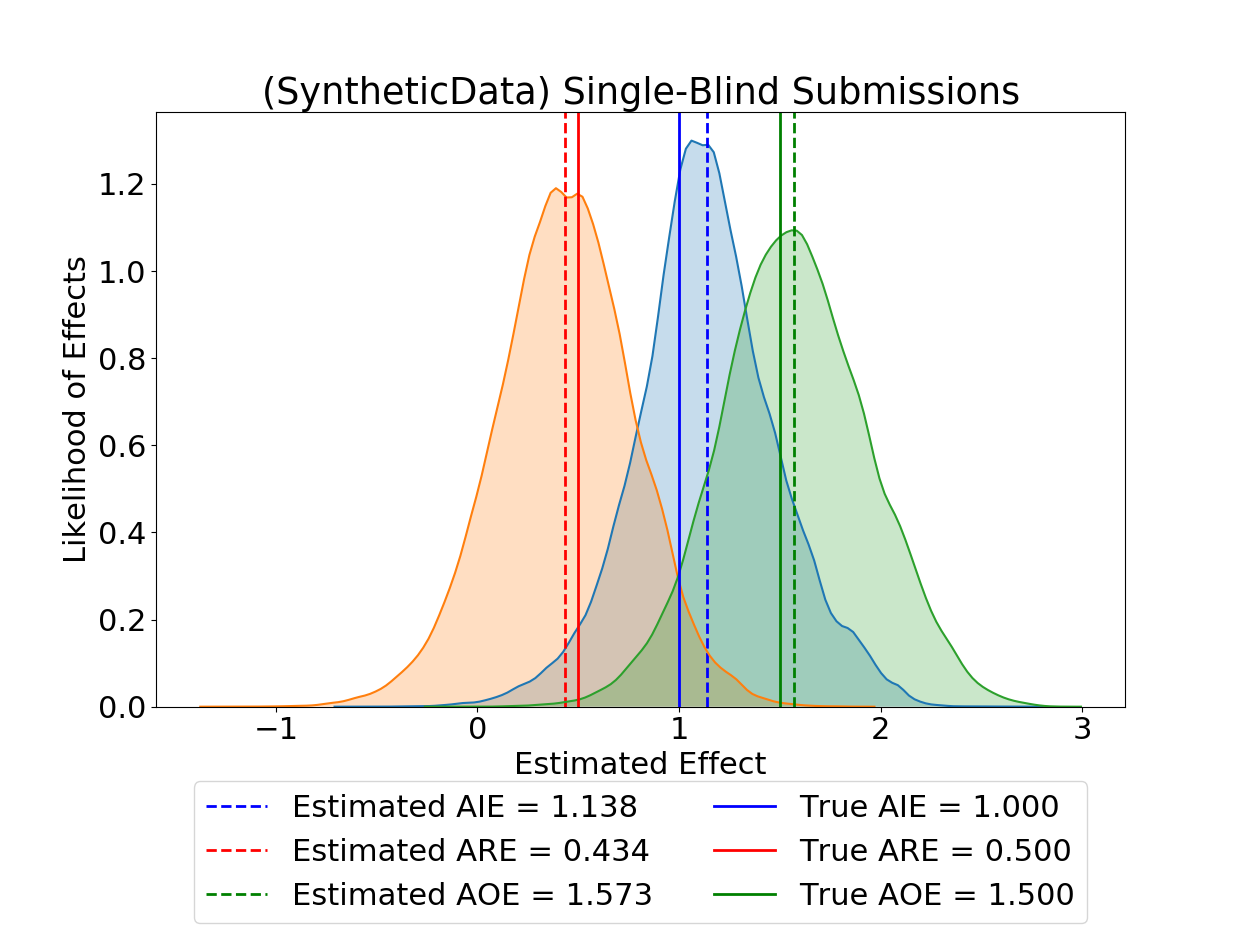}}
    \quad
    \subfloat[]{\includegraphics[width=0.5\textwidth]{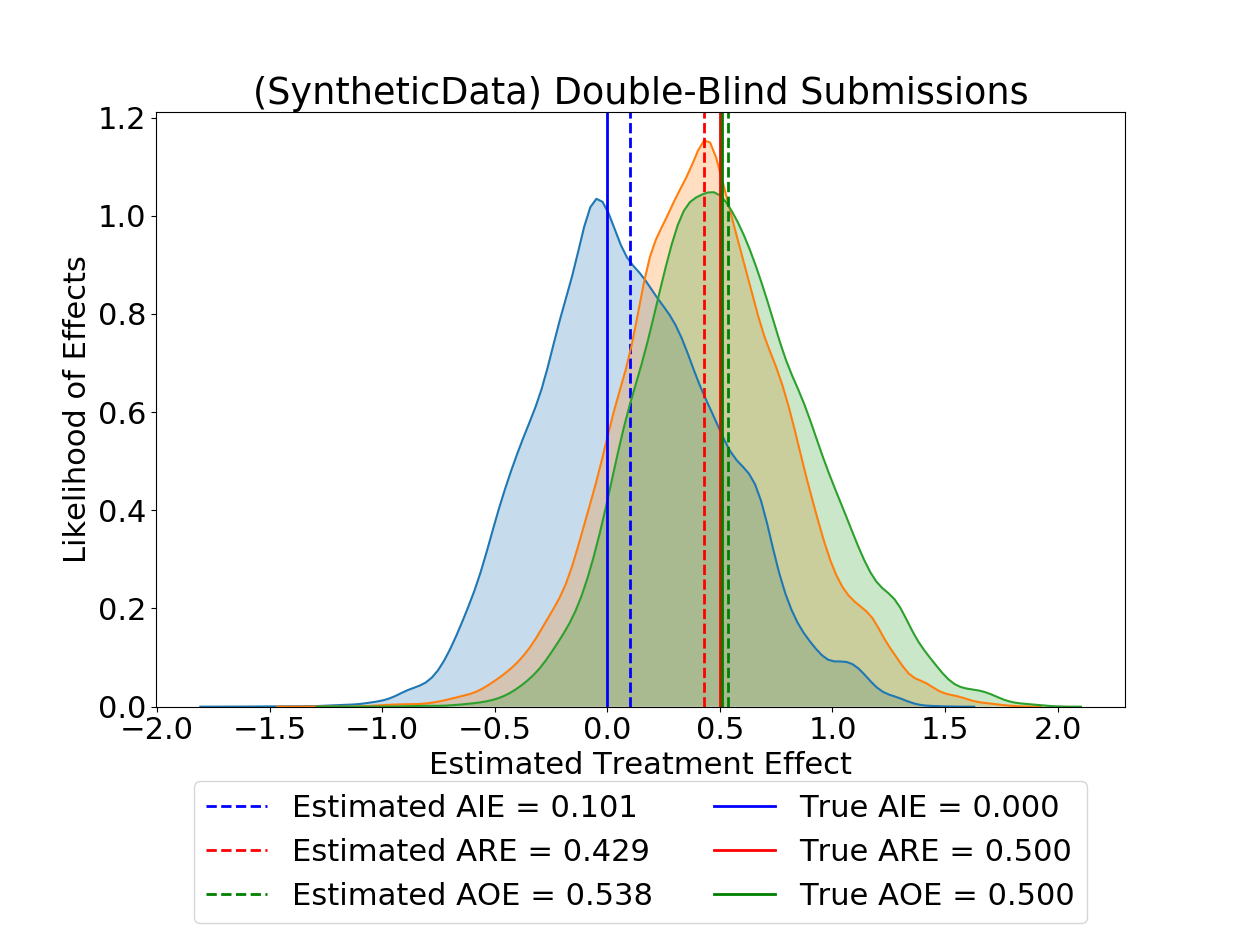}}
    \caption{Relative likelihood of effects and corresponding averages for isolated, relational and overall effects for (a) single-blind venues and (b) double-blind venues.}
    \label{fig:rel_pdf}
\end{figure}
\begin{figure}
    \centering
    \subfloat[]{\includegraphics[width = 0.48\textwidth]{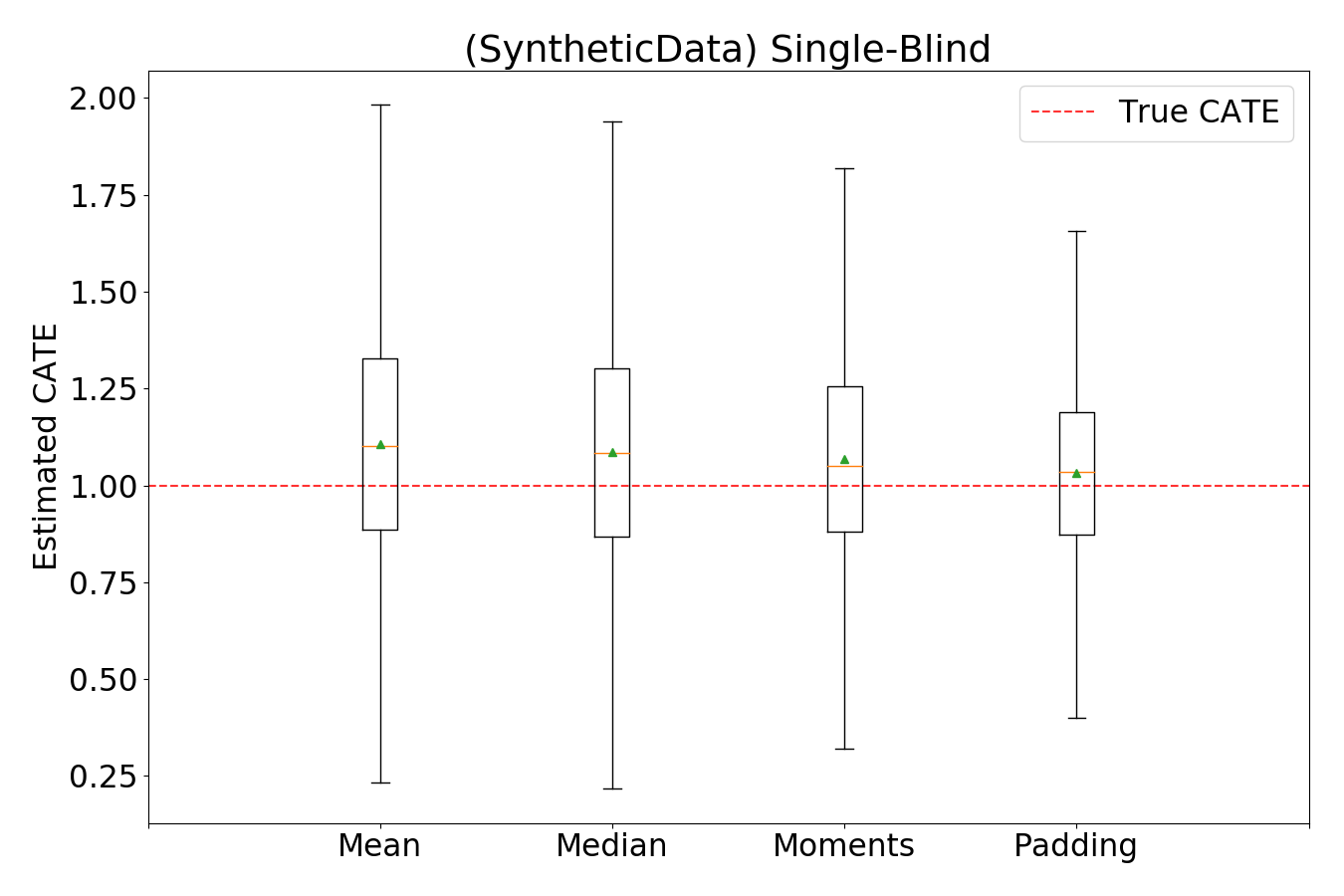}}
    \subfloat[]{\includegraphics[width = 0.48\textwidth]{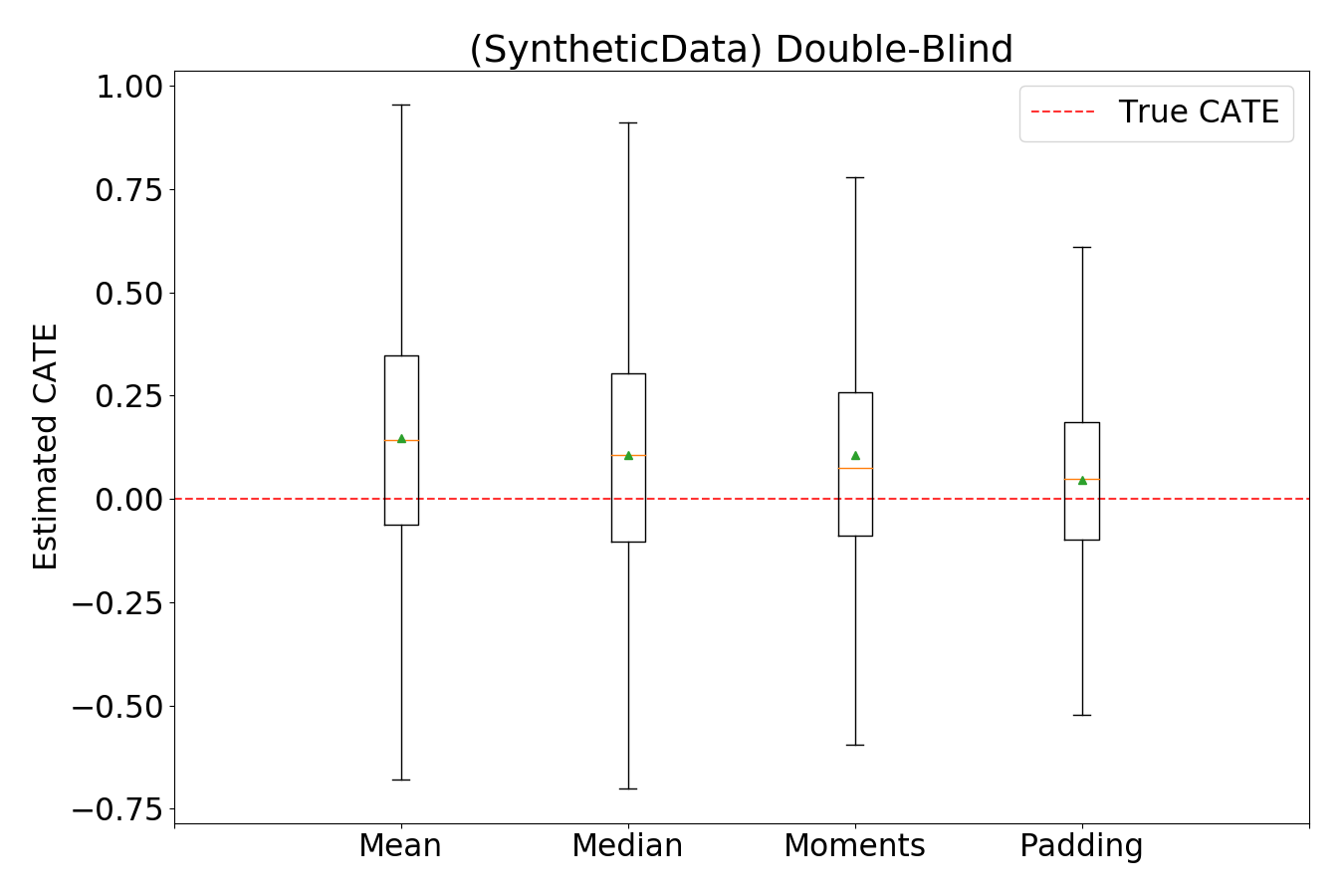}}
    \caption{Comparing the sensitivity of of the quality of query answer (CATE) to different choice of embeddings for (a) Single-Blind submissions and (b) Double-blind submissions.}
    \label{fig:cate_violin}
\end{figure}

%% file: mimic.tex
\cut{
\subsection{\revm{\mimicdata} [To Be Removed]}\label{sec:expt-mimicdata}

\revm{

\paragraph{\bf Causal model and queries.} We are interested in answering causal query if having no insurance has any direct effect on the length of stay, and if having no insurance increases one's chances of mortality. As we are interested in directed effect of self paying on the outcomes, we need to control for the common confounders and adjust for mediators.

\paragraph{\bf Results}  
}

\vspace{-0.2cm}
}

%% file: nis.tex
\cut{
\revm{
 \vspace{-0.5cm}
\subsection{\textsc{Nationwide Inpatient Sample [To be removed]}}\label{sec:expt-nisdata}

\paragraph{\bf Causal Model and Query} 

We studied 
the causal effect of the size of the hospital on the cost of healthcare expressed in (\ref{eq:query_sizebill}). Note that attributes such as illness severity and admission source are confounders. For example, if the source of admission is a transfer from another hospital, both the chances that the patient is moved to a larger hospital and the chances that the patient needs more expensive care increase.

{\bf Results.~} 
}

}

%% file: Conclusion.tex
\section{Conclusions and Future Work} \label{sec:conclusion}
\revm{We introduced the Causal Relational Learning framework for performing causal inference on relational data. This framework allows users to encode background knowledge using a declarative language called \emph{\sys\ (Causal Relational Language)} using simple Datalog-like rules, and ask various complex causal queries on relational data.} \revb{\sys\ is designed for researchers and analysts with a social science, healthcare, academic or legal background who are interested inferring causality from a complex relational data.} \revm{\sys\ adds on to existing causal inference literature by relaxing the \textit{unit-homogeniety assumption} and allowing the confounders, treatment units and outcome units to be of different kinds. We evaluated \sys's completeness and correctness on real-world and synthetic data from academic and healthcare domains. \sys~ is successfully able to recover the treatment effects for complex causal queries that may require multiple joins and aggregates.}

\revm{In future, we aim to extend \sys\ to deal with complex cyclic causal dependencies using stationary distribution of stochastic  processes.}  We plan to study stochastic interventions and complex interventions on relational skeletons, which are assumed to be fixed in this paper.  \revc{We also plan a theoretical and methodological study the functionality of different types of embeddings. We aim to develop principled learning approach for finding efficient embeddings using graph representation learning and graph embedding.}

Recently, it has been shown that causality is foundational to the emerging field of algorithmic fairness \cite{salimi2019capuchin}. In future work we plan to use causal relational learning to  study a causality-based framework for fairness and discrimination in relational domains.

%% file: proof.tex
\pagebreak
\revd{\begin{proof}[Proof of Theorem~\ref{theore:reladj}](Sketch) Since embeddings $\Emb^{\skl}$ are deterministic functions of ground attributes $\FAtt^{\skl}$, they are also random variables, hence we can define the joint probability distribution $\pr(\FAtt^{\skl},\Emb^{\skl})$.
Also note that the parents of an atom in the augmented ground causal diagram $\egcg$ corresponded to the embedded parents of the same node in ground causal diagram $\gcg$.

Since each atomic intervention $do(\treat[\mb x_i]=t_i)$ modifies the augmented ground causal diagram  $\egcg$ by removing the parents of $\treat[\mb x_i]$ from $\egcg$ (implied from the factorization Eq~\ref{eq:rel-fac}), the post intervention distribution $\pr\Big(\FAtt^{\skl}, \Emb^{\skl} \mid \Do\big(\treat[\arbtreat]=\trevec_{\arbtreat}\big)\Big)$ can be obtained from the pre-intervention (observed)  distribution $ \pr(\FAtt^{\skl})$  by removing all factors
	$\pr\Big(\fAtt(\mb x),  \mid \Pa(\fAtt(\mb x)\big)\Big)$, from $ \pr(\FAtt^{\skl},\Emb^{\skl})$ (cf. Eq~\ref{eq:rel-fac}),
	hence we obtain the following: 
{\small
\begin{align}
    \pr\Big(\FAtt^{\skl},\Emb^{\skl} \mid \Do\big(\treat[\arbtreat]=\trevec_{\arbtreat}\big)\Big) =  \frac{\pr(\FAtt^{\skl})}{\prod_{\mb x \in \arbtreat} \pr\Big(\fAtt(\mb x) \mid \Emb^{\fAtt}\big(\Pa'(\fAtt(\mb x)\big)\Big)}  \label{eq:proff_intervetional}
\end{align}
}
The following factorization implied by the chain rule of probability:

{ \small
\begin{align}
 \pr(\FAtt^{\skl}, \Emb^{\skl})  = & 
\pr\Big(\bigcup_{\mb x\in \arbtreat}  \Pa\big(\treat[\mb x]) \Big) \  \pr\Big(\treat[\mb x_1] \mid \bigcup_{\mb x\in \arbtreat} \Pa(\treat[\mb x])\Big)  \\&
 \pr\Big(\treat[\mb x_2] \mid \treat[\mb x_1], \bigcup_{\mb x\in \arbtreat} \Pa(\treat[\mb x])\Big)  \\&
 \nonumber  \hspace{0.5cm} \vdots  \\&
\pr\Big(\treat[\mb x_i] \mid \bigcup^{i-1}_{j=0} \treat[\mb x_j],  \bigcup_{\mb x\in \arbtreat}  \Pa\big(\treat[\mb x])\Big)  \\&
 \hspace{0.5cm} \vdots  \nonumber  \\&
\pr(\overline{\fAtt^{\skl}}, \overline{\Emb^{\skl}} \mid  \bigcup_{\mb x\in \arbtreat}  \Pa\big(\treat[\mb x]), \bigcup_{\mb x\in \arbtreat}  \treat[\mb x])  
\label{eq:proof_fac}
\end{align}
}
where $\overline{\fAtt^{\skl}} \cup \overline{\Emb^{\skl}}$        consist of all ground atoms in $\FAtt^{\skl} \cup \Emb^{\skl}$ except for  $\bigcup_{\mb x\in \arbtreat} \Pa\big(\treat[\mb x])$ and  $\bigcup_{\mb x\in \arbtreat}  \treat[\mb x])  \big\}$.
 The acyclicity of $\egcg$ implies $ \bigcup_{\mb x\in \arbtreat}  \Pa\big(\treat[\mb x])$ and  $\bigcup_{\mb x\in \arbtreat}  \treat[\mb x]$ are disjoint, hence the above factorization is valid.

The following implied from the factorization in Eq.~\ref{eq:proof_fac} and Eq~\ref{eq:proff_intervetional},
{\small
\begin{align} 
\pr\Big(\outc[\mb x']=y \mid \Do\big(\treat[\arbtreat]=\trevec_{\arbtreat}\big)\Big)& = \nonumber\\
 &  \hspace{-3cm}  \sum_{\bigcup_{\mb x\in \arbtreat}  \pa\big(\treat[\mb x]\big)} \
 \pr\big(Y[\mb x']=y \mid \bigcup_{\mb x\in \arbtreat}  \pa\big(\treat[\mb x]), \treat[\arbtreat]=\trevec_{\arbtreat} \big)
 \pr(\bigcup_{\mb x\in \arbtreat}  \pa\big(\treat[\mb x])\big)
   \label{eq:relational:adjustment:parent}
\end{align}
}
Now, given a set $\mb Z \subseteq \Emb^{\skl}$, we can rewrite the RHS of Eq.~\ref{eq:relational:adjustment:parent}  into the following equivalent formula:

{\small
\begin{align} 
\text{RHS}= \nonumber\\
 &  \hspace{-1cm}  \sum_{\bigcup_{\mb x\in \arbtreat} \pa\big(\treat[\mb x]\big)} \   \sum_{\mb z \in \Dom(\mb Z)}
 \pr\big(Y[\mb x']=y \mid \bigcup_{\mb x\in \arbtreat}  \pa\big(\treat[\mb x]), \treat[\arbtreat]=\trevec_{\arbtreat}, \mb Z=z \big) \nonumber \\ & \pr(\mb Z=z \mid \bigcup_{\mb x\in \arbtreat}  \pa\big(\treat[\mb x]), \treat[\arbtreat]=\trevec_{\arbtreat})
 \  \pr(\bigcup_{\mb x\in \arbtreat}  \pa\big(\treat[\mb x])\big) 
 \nonumber \\ & 
   \label{eq:proof-rewritten}
\end{align}
}

Now from the conditional independence in Eq.\ref{eq:raf_idep} and the conditional independence statement  $\treat[\mb x] \indep \mb Z , \bigcup_{\mb x\in \arbtreat}  \pa\big(\treat[\mb x]) \mid \pa(\treat[x])$ 
for each $\mb x \in \arbtreat$, Eq~\ref{eq:proof-rewritten} can be simplified as follows:
{\small
\begin{align} 
\text{RHS}= \nonumber\\
 &  \hspace{-1cm}  \sum_{\bigcup_{\mb x\in \arbtreat} \pa\big(\treat[\mb x]\big)} \   \sum_{\mb z \in \Dom(\mb Z)}
 \pr\big(Y[\mb x']=y \mid \bigcup_{\mb x\in \arbtreat}  \pa\big(\treat[\mb x]), \treat[\arbtreat]=\trevec_{\arbtreat}, \mb Z=z \big) \nonumber \\ & \pr(\mb Z=z \mid \bigcup_{\mb x\in \arbtreat}  \pa\big(\treat[\mb x]), \treat[\arbtreat]=\trevec_{\arbtreat})
 \  \pr(\bigcup_{\mb x\in \arbtreat}  \pa\big(\treat[\mb x])\big) 
 \nonumber \\ &=  
  \sum_{\bigcup_{\mb x\in \arbtreat} \pa\big(\treat[\mb x]\big)} \   \sum_{\mb z \in \Dom(\mb Z)}
 \pr\big(Y[\mb x']=y \mid  \treat[\arbtreat]=\trevec_{\arbtreat}, \mb Z=z \big) \nonumber \\ & \pr(\mb Z=z \mid \bigcup_{\mb x\in \arbtreat}  \pa\big(\treat[\mb x]))
 \  \pr(\bigcup_{\mb x\in \arbtreat}  \pa\big(\treat[\mb x])\big) 
 \nonumber \\ & 
 = \sum_{\mb z \in \Dom(\mb Z)} \
 \pr\big(Y[\mb x]=y \mid \mb Z= \mb z,\treat[\arbtreat]=\trevec_{\arbtreat} \big)
 \pr(\mb Z=\mb z) \label{eq:lasteq}
\end{align}
}
Taking expectation from both sides of Eq.~\ref{eq:lasteq} completes the proof. 
\end{proof}
}